\documentclass[12pt]{article}

\usepackage{fullpage}

\usepackage[pdftex]{graphicx}
\usepackage{amsfonts}
\usepackage{amscd}
\usepackage{indentfirst}
\usepackage{amssymb,amsthm, amsgen, amstext, amsbsy, amsopn,yfonts}
\usepackage{xcolor}
\usepackage{amsmath}
\usepackage{mathrsfs}
\usepackage{ulem} %

\newtheorem{thm}{Theorem}[section]

\newtheorem{prop}[thm]{Proposition}

\newtheorem{rem}[thm]{Remark}

\newcommand{\J}{{\mathbb{J}}}

\newcommand{\R}{{\mathbb{R}}}

\newcommand{\cP}{{\mathcal{P}}}

\newcommand{\bal}{\begin{aligned}}
\newcommand{\enbal}{\end{aligned}}
\newcommand{\be}{\begin{equation}}
\newcommand{\ee}{\end{equation}}
\newcommand{\farc}{\frac}
\newcommand{\pdr}[2]{\frac{\partial{#1}}{\partial{#2}}}

\newcommand{\Div}{\nabla\!\!\!\!\nabla}
\newcommand{\Rm}{{\mathbb R}}

\begin{document}

\title{Geometric aspects of a spin chain}

\author{\textsc Michael Entov$^{1}$,\ Leonid Polterovich$^{2}$,\ Lenya Ryzhik$^{3}$}

\footnotetext[1]{Partially supported by the Israel Science Foundation grant 2033/23.}

\footnotetext[2]{ Partially supported by the Israel Science Foundation grant 1102/20.}

\footnotetext[3]{ Partially supported by NSF grants DMS-1910023 and DMS-2205497 and by ONR grant N00014-22-1-2174. }

\date{\today}

\maketitle

\begin{abstract}
We discuss non-equilibrium thermodynamics of the mean field Ising model from a geometric perspective, focusing on the thermodynamic limit. When the number of spins is finite, the Gibbs equilibria form a smooth Legendrian submanifold in the thermodynamic phase space  whose points describe the stable macroscopic states of the system. We describe the convergence of these smooth Legendrian submanifolds, as the number of spins goes to infinity, to a singular Legendrian submanifold, admitting an analytic continuation that contains both the stable and metastable states.  We also discuss the relaxation to a Gibbs equilibrium when the physical parameters are changed abruptly. The relaxation is defined via the gradient flow of the free energy with respect to the Wasserstein metric on microscopic states, that is, in the geometric language, via the gradient flow of the generating function of the equilibrium Legendrian with respect to the ghost variables. This leads to a discrete Fokker-Planck equation when the number of spins is finite. We show that in the thermodynamic limit this description is closely related to the seminal model of relaxation proposed by Glauber. Finally, we find a special range of parameters where such relaxation happens instantaneously, along the Reeb chords connecting the initial and the terminal Legendrian submanifolds.
\end{abstract}


\section{Introduction}

Contact geometry provides a useful language for certain facets of
thermodynamics, and there is an extensive body of previous work discussing the
relation between the two subjects -- see e.g.  \cite{BLN,EP,Goto-JMP2015,Goto,GLP,Grmela,Haslach,LimOh,V}.
In this paper we consider an aspect of this connection that, to the best of our knowledge, has not
been discussed in the literature before -- the behavior of the geometric constructions describing
thermodynamic systems and thermodynamic processes in
the thermodynamic limit. To be concrete, we study this question in detail for the
example of the Curie-Weiss magnet consisting of $N$
pair-wise interacting spins \cite{FV}, as $N\to +\infty$.

Let us briefly discuss the main results of the present paper.

\bigskip
\noindent
{\bf Convergence of equilibrium smooth Legendrian submanifolds to a singular Legendrian submanifold in the thermodynamic limit.}
It is well-known (see e.g.~\cite{GrO,Her}) that equilibrium states of thermodynamic
systems lie on the Legendrian submanifolds of the thermodynamic phase spaces $\R^{2n+1}$, equipped
with the standard contact structure. We show that smooth Legendrian submanifolds describing
thermodynamic equilibria for a finite system may converge to a singular Legendrian submanifold in the thermodynamic limit. Interestingly enough, this kind of convergence attracts now a lot of attention in symplectic and contact topology in the study of the completion of the space of Legendrian submanifolds (see Remark \ref{rem-conv} below) but there
are few natural examples of such singular limits. We give
a precise geometric description of such a convergence in the basic example of the Curie-Weiss magnet
-- see Theorem~\ref{thm-limit}.

\bigskip
\noindent
{\bf Non-equilibrium thermodynamics, gradient dynamics and the thermodynamic limit.}
We use contact geometry in order to study the non-equilibrium thermodynamics for the Curie-Weiss magnet consisting of $N$ pair-wise interacting spins. Our general strategy is related to the version of the GENERIC approach to non-equilibrium thermodynamics introduced in \cite{GrO}. Namely, as in \cite{GrO}, we study a generating function for the Legendrian submanifold describing
thermodynamic equilibria, with the so-called ghost variables given by probability measures on the space of microscopic states of the thermodynamic system. The thermodynamic relaxation process is then described by the gradient dynamics of the generating function
on the space of such measures. However, unlike in \cite{GrO}, our main interest is in the thermodynamic limit of such gradient dynamics
in the contact geometric context.

More specifically, the starting point of our consideration is the observation that we learnt from \cite{LimOh} that the generating function of the equilibrium Legendrian submanifold is given by the free energy of the system that we will denote by $z$.
In the setting of the Curie-Weiss magnet consisting of $N$
pair-wise interacting spins, the generating function, that we denote by~$\Psi (q,\rho)$,
depends on the external
magnetic field $q$ as the position variable and on
a probability density $\rho$ on the space of the microscopic states of the system as a so-called ghost variable.
The momentum variable $p$ is the generalized thermodynamic pressure. The equilibrium Legendrian submanifold $\Lambda$
lies in the thermodynamics phase space $\R^3 (p,q,z)$ equipped with the standard contact form $dz-pdq$.
For a finite $N$, a thermodynamic relaxation of the system to an equilibrium -- for instance, after
an instant change of the temperature -- is modeled by the aforementioned gradient flow
of the generating function $\Psi (q,\rho)$ with respect to~$\rho$ and with respect to a
Wasserstein distance on the space of the probability densities
on the finite state space. It turns out that this evolution takes the form of
a discrete Fokker-Planck equation. This is in agreement with the observation in
\cite{JKO,Maas,Otto}
that the dynamics defined by Fokker-Planck equations can be presented as the
gradient flow  with respect to a Wasserstein distance.
The long time limit of the flow is a point on the
equilibrium Legendrian submanifold  corresponding to the new Gibbs equilibrium.
We focus on the thermodynamic limit of this gradient dynamics, as $N\to +\infty$ -- see
Section~\ref{subsec-grad-flow-wasserstein-fokker-planck}.
For the Curie-Weiss magnet, the dynamics defined by the discrete Fokker-Planck equation, or by the Wasserstein distance gradient flow, turns out to be closely
related to the Glauber dynamics,~\cite{BH,Glauber}, a well-known Markov process describing relaxation of the magnet. We explain the connection between the gradient flow dynamics and the Glauber dynamics and describe the behavior of this connection
in the thermodynamic limit -- see Theorem \ref{thm-drifts} for a
precise formulation.

\bigskip
\noindent
{\bf Basins of attraction.}  Our model enables one to describe non-equilibrium dynamics
with the initial conditions satisfying certain constraints. To describe these constraints, we consider,
 for a finite $N$, the set $\mathcal{B}_N\subset\R^3$ formed by the points of the thermodynamic phase space
$\R^3 (p,q,z)$ that can be represented, using $\Psi$, by measures $\rho$ that are {\it not} necessarily Gibbs (that is,
do {\sl not} necessarily correspond to a thermodynamic equilibrium state)
but do evolve to a Gibbs measure with the gradient flow mentioned above. Namely,
\begin{equation}\label{eq-basin-intro}
\mathcal{B}_N := \bigg\{ (p_0,q_0,z_0)
 \in \R^3 \; \Big{|}   \; \exists \rho \;:\; p_0= \frac{\partial \Psi}{\partial q}(q_0,\rho),  \; z= \Psi(q_0,\rho)\bigg\}.
 \end{equation}
We call such a set $\mathcal{B}_N$ {\it the basin of attraction} of the equilibrium Legendrian $\Lambda$.
We give a precise description of
$\mathcal{B}_N$ (see Proposition~\ref{prop-basin-descr}) and prove a result on
the limiting shape of~$\mathcal{B}_N$ as $N\to +\infty$ -- see Theorem~\ref{thm-phys}.

\bigskip
\noindent{\bf A connection to contact topology:} In Section \ref{subsec-chords} we discuss the
relaxation of the Curie-Weiss magnet after a sudden change of its temperature and the external magnetic field. This process is described as an instant jump of free energy and the relaxation is modeled by the gradient Fokker-Planck evolution. We focus on such a relaxation in the thermodynamic limit and consider the special case, occurring in a certain range of parameters and for a special value of the magnetic field, where the relaxation is instant: after the free energy jump, the system arrives at the microscopic equilibrium of the Fokker-Planck evolution. Our observation is that on the macroscopic level, this special case corresponds to a Reeb chord, i.e., a segment parallel to $z$-axis, joining the initial and the terminal equilibrium Legendrian submanifolds. Such chords are central objects of study in modern contact topology (see e.g.~\cite{EP} and the references therein).
While the existence of the instantaneous jumps along  the Reeb chords in our case is relatively simple to see, it may not be so in more sophisticated thermodynamic models. We expect that topological methods will be useful for this task. We should emphasize that such instantaneous jumps
exist only in the thermodynamic limit
and not for a finite system.


\medskip
\noindent
{\bf Organization of the paper:}
\\
In Section~\ref{sec-setting}, we first describe the basic notions of thermodynamics in Section~\ref{sec:glossary}
and then recall the notion of a Legendrian submanifiold in contact geometry in Section~\ref{sec:contact}. We also explain there the
interpretation of a thermodynamic equilibrium in terms of a Legendrian submanifold. The finite Curie-Weiss model is introduced in Section~\ref{finite-CW}.
Its thermodynamic limit is considered in Section~\ref{subsec-limit}, in terms of the convergence of the properly rescaled Legendrian submanifolds.
This convergence as well as the limiting singular Legendrian submanifold are described in
Theorem \ref{thm-limit}.

Section~\ref{sec:noneq} discusses the non-equilibrium dynamics and relaxation to a thermodynamic equilibrium in terms of the gradient flow
in the Wasserstein distance on the space of probability measures on the state space.
Recall that these measures serve as the ghost variables for the equilibrium Legendrian submanifold, and that the resulting evolution is described by the (discrete) Fokker-Planck equation. Next we focus on the thermodynamic limit of the gradient dynamics.
In particular, we establish some similarities between the gradient flow approach and the Glauber dynamics, see Theorem~\ref{thm-drifts}. This theorem is proved in Section~\ref{sec-appendix}.

It should be mentioned that the model of the non-equilibrium thermodynamics presented
in this paper is defined in a certain part of the thermodynamic phase space.
In Sections~\ref{sec:basin} and \ref{subsec-basin} we elaborate this and detect the basin of attraction for the Curie-Weiss model.

One of the themes of this paper is the interplay between gradient (Fokker-Planck) dynamics and contact dynamics within the framework of non-equilibrium thermodynamics. In Section \ref{subsec-chords} we relate the gradient dynamics to the Reeb chords  connecting initial and terminal equilibrium Legendrian submanifolds. In Section \ref{sec-cont} we compare the Fokker-Planck drift with contact Hamiltonian dynamics for a toy thermodynamic system.

We finish with a short summary of results in Section~\ref{sec:conlcusion}. The appendices in Sections~ \ref{App-2}-\ref{sec-entropy}
contain some computations used in the main text.

{\bf Acknowledgment.} We express our deep gratitude to Shin-itiro Goto for numerous illuminating
discussions, for useful comments on the manuscript, and for help with figures.
 Figures 1,6,7 were produced by the authors with the use of ChatGPT 4.
 Preliminary results of this paper were presented by L.P. at a conference ``From smooth to $C^0$ symplectic geometry: topological aspects and dynamical implications" held in July, 2023 at the Centre International de Rencontres Math\'{e}matiques (Marseille, France). The recording of this talk, called ``Contact topology and thermodynamics", is available on YouTube. L.P. thanks the organizers and the CIRM for this opportunity.

\section{Contact geometry and thermodynamic equilibria}\label{sec-setting}

\subsection{Thermodynamic equilibria}\label{sec:glossary}

{We start with a brief description of thermodynamic equilibria in a way that will be convenient for their geometric interpretation in the following section.}

{The space of the microscopic states of a thermodynamic system is modeled by a manifold~$M$, formed by the microstates, together with a measure  $\mu$. We model a macroscopic thermodynamic state
by a probability measure on $M$ that has a smooth positive density~$\rho>0$ with respect to $\mu$,
so that~$\int_M \rho d\mu = 1$. Such a density describes the distribution of a physical quantity (e.g. spin) over the microstates.} \color{black}
We  denote by $\mathcal{P}$ the set of all such densities.
{The positivity assumption on $\rho$ allows
us  to define the entropy of a macroscopic
state as}
\[
S(\rho) = -\int_M \rho \ln \rho~d\mu\;\;,~~ {\rho\in{\mathcal P}}.
\]

The free energy of the system is
\begin{equation}\label{eq-Phi-rho}
\Phi(q,\rho)=
-\beta^{-1}S(\rho) + \int_M H(q,m) \rho(m)d\mu(m),~~q\in\Rm^n,~\rho\in\cP.
\end{equation}
Here, $q\in\Rm^n$ is an external {physical} parameter,  $H(q,m)$  is the Hamiltonian that
measures the energy of a microscopic state, and
$\beta>0$ is the inverse temperature of the system.

An equilibrium of a thermodynamic system for fixed inverse temperature $\beta >0$,
and  external parameter $q\in\Rm^n$ is  a probability density in $\cP$
that is a minimizer
of the variational problem
\begin{equation}\label{eq-var}
{\min_{\rho\in{\mathcal P}}\Phi(q,\rho).}
\end{equation}
Note that for every $q$ fixed,
the functional $\Phi(q,\cdot)$ is strictly convex on $\mathcal{P}$.
Hence,
it has only one critical point known as the
Gibbs distribution. To find it, note that
at an equilibrium we have
\be\label{nov2002}
\frac{\partial \Phi}{\partial \rho}(q,\rho) =0,
\ee
from which, together with the definitions of $S(\rho)$ and $\Phi(q,\rho)$,  we deduce that
\[
\ln \rho +\beta H = \text{const}.
\]
This readily yields the Gibbs distribution
\[
\rho_G(q,m)= \frac{e^{-\beta H(q,m)}}{\mathcal{Z}(q)}.
\]
Here, $\mathcal{Z}(q) = \int_M e^{-\beta H(q,m)} d\mu(m)$ is   the {normalization constant known as the}
partition function.
{An important role in the considerations below will be played by the free energy of the Gibbs distribution}
\begin{equation}\label{eq-phiH}
\phi_H(q) = \Phi(q,\rho_G(q,\cdot)).
\end{equation}

\begin{prop}\label{prop-secondder}
Assume that the Hamiltonian $H(q,m)$ is linear 
in $q$ for every $m\in M$. Then
\begin{equation}\label{eq-phi-form}
\phi_H(q) = -\beta^{-1}\ln \mathcal{Z}(q)\;,
\end{equation}
and
\begin{equation}\label{eq-phi-der}
\phi''_H(q) = \beta\Big( \Big(\int \frac{\partial H}{\partial q} \rho_G d \mu \Big)^2 -
\int\Big(\frac{\partial H}{\partial q}\Big)^2 \rho_G d\mu \Big) \;.
\end{equation}
In particular, $\phi_H(q)$ is a smooth concave function.
\end{prop}

\begin{proof} Formulas \eqref{eq-phi-form} and \eqref{eq-phi-der} are obtained
by a direct calculation that
uses the linearity of~$H(q,m)$ in $q$. The concavity of $\phi_H(q)$ follows from
\eqref{eq-phi-der} by the Cauchy-Schwartz inequality.
\end{proof}

{In addition to the entropy and the free energy, another important collection of functionals of a given macroscopic state $\rho\in\cP$} are
the generalized pressures~\cite{Be}
\begin{equation}
\label{eq-free-en-gf}
p_i{(q,\rho)} = -\frac{\partial \Phi}{\partial q_i}  {(q,\rho)}
= \int_M \frac{\partial H}{\partial q_i} {(q,m)} \rho {(m)} d\mu({m})\;, \; i=1,\dots , n\;.
\end{equation}
{Their role will become clear in the next section.}

\subsection{{Thermodynamic equilibria as Legendrian submanifolds}}\label{sec:contact}

The basic notions introduced above admit the following geometric interpretation. {Given a macroscopic state $\rho\in\cP$ we will  consider
the point $(p(q,\rho),q,-\Phi(q,\rho))$ as an element of}
the thermodynamic phase space
$\J^1\R^n:=T^*\R^n \times \R$ equipped with the coordinates $(p,q,z)$. Here,~$q \in \R^n$ is an external parameter, $p \in (\R^n)^*$ is the corresponding generalized pressure, and $z \in \R$ is the free energy taken with the opposite sign.
The space~$\J^1\R^n$ is equipped with the contact 1-form
$$\lambda = dz - \sum_{i=1}^n p_idq_i\;,$$
called the Gibbs form. Recall that being contact means that $\lambda \wedge (d\lambda)^n$ is a volume form.

An $n$-dimensional submanifold $\Lambda \subset \J^1\R^n$ is called Legendrian
 if $\lambda$ vanishes on $T\Lambda$. The simplest example of a Legendrian submanifold is  the
1-jet of a smooth function $f: \R^n \to \R$,
\begin{equation}\label{eq-lambdaef}
\Lambda_{f} = \Big\{ (p,q,z)\in{\J^1\R^n}\Big|~z= f(q),~ p = \frac{\partial f}{\partial q}(q)\Big\}\;.
\end{equation}
This generalizes as follows. Given a {smooth}
function $$\Psi: \R^n \times E \to \R\;,$$
where $E$ is a space of auxiliary variables $\xi$,
the set \begin{equation}\label{eq-gf}
\Lambda_{\Psi} = \Big\{\hbox{$(p,q,z)\in{\J^1\R^n}$ $\Big|$ $\exists\xi\in E$ such that }
\frac{\partial \Psi}{\partial \xi} (q,\xi) = 0,~
p= \frac{\partial \Psi}{\partial q} (q,\xi),~ z= \Psi(q,\xi)\Big\}
\end{equation}
is a (possibly singular) Legendrian submanifold for a generic $\Psi$. The function $\Psi$
is called a generating function of $\Lambda_{\Psi}$, and the variables $\xi$ are called
the ghost variables.

With this language, in light of formulas (\ref{nov2002}) and \eqref{eq-free-en-gf},
an equilibrium of a thermodynamic system discussed in
the previous section can be described {as a point on} a Legendrian submanifold
\be\label{nov2106}
\Lambda = \Big\{(p,q,z) \in \J^1\R^n \; \Big{|}   \; \exists \rho \in \mathcal{P} \;:\; \frac{\partial \Phi}{\partial \rho}(q,\rho) =0, \; p = - \frac{\partial \Phi}{\partial q}(q,\rho), \; z= -\Phi(q,\rho)\Big\}\;.
\ee
{That is,
the macroscopic densities $\rho\in\cP$ play the role of the ghost variables, and }
the (minus) free energy $-\Phi$ is a generating function of $\Lambda$ (this observation also appears in \cite{LimOh}).
{The ghost variables that generate the  points on the Legendrian submanifold are precisely the Gibbs equilibria.}

 When the free energy $\Phi(q,\rho)$ is convex in $\rho$, so that it has a unique critical point, the equilibrium Legendrian (\ref{nov2106}) can be written
in the form of (\ref{eq-lambdaef})
\begin{equation}\label{nov2108}
\Lambda = \Big\{  (p,q,z)\in{\J^1\R^n}\Big|~z= f(q),~ p = \frac{\partial f}{\partial q}(q)\Big\}\;,
\end{equation}
with $f(q)=-\phi_H(q)$.

Note that the fact that the equilibrium set is Legendrian, i.e. satisfies
$dz-pdq=0$, can be considered  as a manifestation of the Gibbs relation,
which combines the first and the second laws of thermodynamics.

\subsection{A finite Curie-Weiss magnet} \label{finite-CW}

{The observation that thermodynamic equilibria can be expressed in terms of Legendrian
submanifolds is not knew and goes back at least to~\cite{GrO}, while the idea that macroscopic
state densities play the role of the ghost variables has also appeared in~\cite{LimOh}.
Our interest here is two-fold:
first, to understand what happens to this construction in the thermodynamic limits, and, second, to suggest a possible
version of non-equilibrium dynamics in that context.}

Our central example is the basic Curie-Weiss model  of $N$ spins with
values~$\sigma= +1$ (spin up), and $\sigma=-1$ (spin down) in the presence of an external
magnetic field $q \in \R$.  The mean spin
$$
{m:=}\frac{1}{N} \sum_{j=1}^N \sigma_j$$
takes values in the set
$$M_N := \{-1,-1+2/N, \dots, 1-2/N, 1\}\;.$$

The space of microscopic states $\Upsilon_N$ consists of
all spin configurations $\{\sigma_j\}, j=1,\dots, N$, with $\sigma_j=\pm1$. The
 Curie-Weiss Hamiltonian is
\begin{equation}\label{eq-hamnoneff}
h(\sigma) = -N(qm+bm^2/2)\;.
\end{equation}
Here, the constant $b >0$ is determined by the magnetic properties of the material.
The sign $b >0$ means that  the spins tend to align when the external field is applied,
which is true for a ferromagnetic material.

An efficient tool for studying the Curie-Weiss magnet is a reduction from the space $\Upsilon_N$ of~$2^N$ spin configurations $\{\sigma_j\}, j=1,\dots, N$ to a (much smaller) space $M_N$ of the cardinality $N+1$ (see, e.g., \cite{BH,FV}).
It turns out that after this reduction, the effective
Hamiltonian $H_N$ of the Curie-Weiss magnet with $N$ spins is given by
\begin{equation} \label{eq-ham-red}
H_N(q,m) = N F_{b,\beta}(q,m) +\beta^{-1} r_N(m)\;,~~m\in M_N\;,
\end{equation}
where
\begin{equation}\label{eq-redCW}
F_{b,\beta}(q,m) = -qm - \farc{bm^2}{2} + \frac{1+m}{2\beta} \ln \frac{1+m}{2} +
\frac{1-m}{2\beta}\ln \frac{1-m}{2}\;.
\end{equation}
The remainder $r_N(m)$ in \eqref{eq-ham-red} is
independent of $q$ and, for $m \in (-1,1)$,  has the form
\begin{equation} \label{eq-remainder}
r_N(m) = \frac{1}{2}\ln (8\pi N)- \frac{1}{2}\ln\frac{4}{1-m^2} - \frac{1}{12N}\left(1-\frac{4}{1-m^2}\right) + O(1/N^2)\;.
\end{equation}
It should be mentioned that the quadratic term $ -qm - bm^2/2 $ in \eqref{eq-redCW}
is just $N^{-1}h$, with $h$ as in \eqref{eq-hamnoneff}. It is responsible for the interaction of spins. The logarithmic terms in \eqref{eq-redCW} and the error term $r_N(m)$
appear when one recalculates the entropy (entering into the expression~\eqref{eq-Phi-rho} for the free energy) of probability distributions on $\Upsilon_N$ to their push-forwards to $M_N$. We refer to Section \ref{sec-entropy} for  the details of this calculation.

From now on, we work on the space $M_N$.
We view it as a $0$-dimensional manifold equipped with the measure $\mu_N(m)= 2/N$ for
every $m \in M_N$. The set of all positive probability densities $\mathcal{P}_N$
on $M_N$ is the  simplex 
\be\label{nov2112}
\Big\{(\rho_1,\dots,\rho_{N+1}) \in \R^{N+1}\;:\; \sum_i \rho_i = \frac{N}{2}\;,
\rho_i >0 \; ~~\forall i\Big\}\;.
\ee

Let us mention also that ${\partial H_N}/{\partial q} = - Nm$,
and hence the generalized pressure $p$ is given by
\begin{equation}\label{eq-pres-cw}
p = N \langle m \rangle_{\rho}\;,
\end{equation}
where $\langle \cdot \rangle_{\rho}$ stands for the expectation with respect to the measure
$\rho d\mu$. The expectation of the mean spin, $\langle m \rangle_{\rho}$, is called the magnetization.

We will denote by $\Lambda_N$ the equilibrium
Legendrian of the finite Curie-Weiss model with~$N$ spins, as given by (\ref{nov2108}), with $f(q)=-\phi_{H_N}(q)$:
\begin{equation}\label{jul1202}
\Lambda_N = \Big\{  (p,q,z)\in{\J^1\R^n}\Big|~z= f(q),~ p = \frac{\partial f}{\partial q}(q)\Big\}\;.
\end{equation}
{In the next section, we will investigate the thermodynamic limit $N\to+\infty$ in terms of the limit of that submanifold.}

\subsection{The thermodynamic limit $N \to +\infty$}\label{subsec-limit}

{We now describe the thermodynamic limit  $N \to +\infty$ of the equilibrium
Legendrian $\Lambda_N$ of the Curie-Weiss model, given by (\ref{jul1202})}.
{To this end, we first introduce
the submanifold}
\be\label{nov2004}
\Lambda_\infty :=
\left\{ (p,q,z)\in{\J^1\R^n}\Big|~ \frac{\partial F_{b,\beta}}{\partial p} (q,p)=0,~
~z= - F_{b,\beta}(q,p)\right\} \subset \J^1\R^1\;,
\ee
It is straightforward to check
that with $F_{b,\beta}$ given by \eqref{eq-redCW},
the submanifold $\Lambda_\infty$
is Legendrian. In the infinite model context, the variable $p$ is the limit of the rescaled pressure $p\to p/N$  of
the Curie-Weiss model with $N$ spins.
By \eqref{eq-pres-cw}, it can be identified with the  limit of the magnetizations
$\langle m \rangle_\rho$.  Because of that connection, we call the
Legendrian submanifold $\Lambda_\infty$
the  Curie-Weiss Legendrian.

Note that the equation $\frac{\partial F_{b,\beta}}{\partial p} (q,p)=0$  in the definition
of~$\Lambda_\infty$ has the form
\be\label{nov2102}
p= \tanh \beta(q+bp),
\ee
which is the standard mean-field equation for an infinite Curie-Weiss model \cite{FV,Mussardo}.
In what follows, we focus on the regime $b\beta > 1$ -- under this assumption the system
exhibits a phase transition. The graph of $F_{b,\beta}$ in this case has the form as in Fig.~1,
{with two local minima and a local maximum}.
\begin{figure}[htb]
\centering
\includegraphics[width=0.6\textwidth]{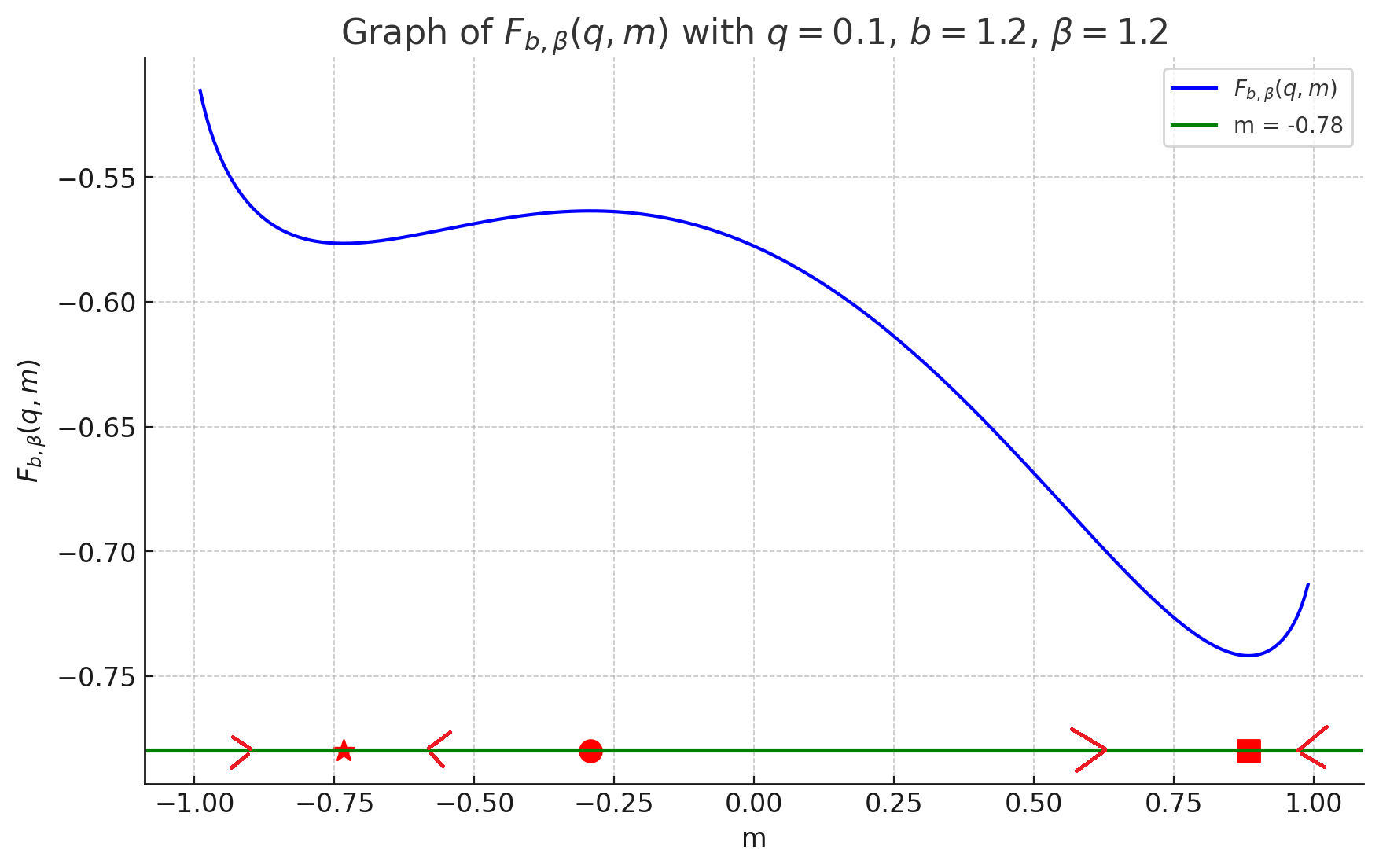}
\caption{Graph of $F_{b,\beta}$;
arrows indicate the (drift) vector field expressing relaxation process;
$\bullet$ denotes unstable  point, $\star$ metastable point, $\blacksquare$ stable point.}
\end{figure}
The projections of the submanifold $\Lambda_\infty$ onto the $(q,p)$- and $(q,z)$-planes {in that regime}
are
depicted in Fig.~2 and Fig.~3, respectively.
\begin{figure}[!tbp]
  \centering
  \begin{minipage}[b]{0.45\textwidth}
    \includegraphics[width=0.7\textwidth]{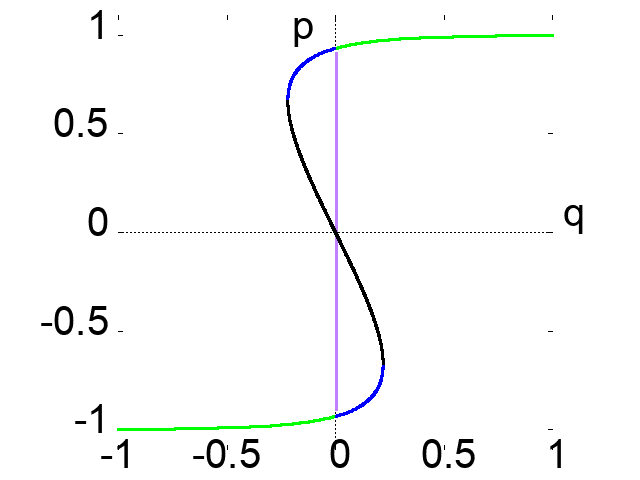}
    \caption{Curie-Weiss Legendrian: Lagrangian projection}
  \end{minipage}
  \hfill
  \begin{minipage}[b]{0.45\textwidth}
  \includegraphics[width=0.7\textwidth]{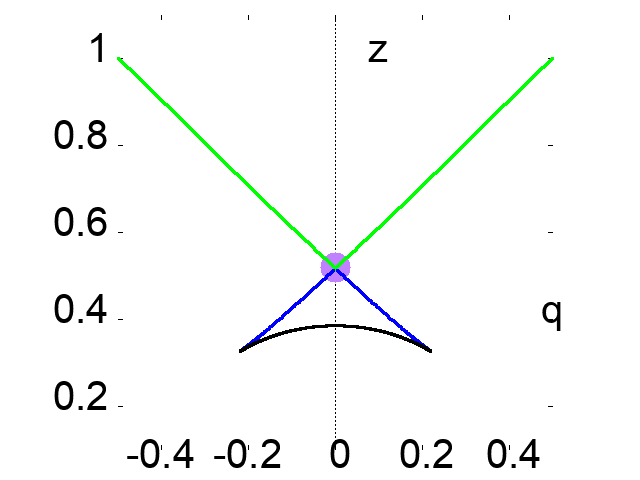}
    \caption{Curie-Weiss Legendrian: front projection}
  \end{minipage}
\end{figure}

The part of $\Lambda_\infty$ presented in green on the front projection (see Fig.~2)  corresponds to the absolute minimum of $F_{b,\beta}$
and is called
{stable}, the blue part corresponding to the second (larger) minimum is called metastable, and the black part corresponding to the maximum is called
{unstable}.
These are the three critical points of $F_{b,\beta}$
seen in Fig.~1.

The stable part of $\Lambda_\infty$ in $(z,q)$-plane (the green curve in Fig.~3) is given by the equation
\begin{equation}
\label{eq-stablepart}
z= f(q) : = \max_m \left(-F_{b,\beta}(q,m)\right)\;.
\end{equation}
Note that the function $f$ is smooth away from the point $q=0$. When $b\beta >1$, the function~$f$ is continuous, albeit non-smooth at $0$. This reflects the fact that the infinite Curie-Weiss model exhibits a phase transition of the first type for such values of $b$ and $\beta$. A direct calculation shows that $f$ is convex.

\begin{rem}\label{rem-free-energy}{\rm Denote by $\mathcal{Z}_N$ the partition function of the Curie-Weiss model with $N$ spins. By \cite[Theorem 2.8]{FV}, after adjusting notation,
$$f(q) = \beta^{-1}\lim_{N\to +\infty} \frac{1}{N}\mathcal{Z}_N\;.$$
In many sources (see e.g. formula (3.1.6) in \cite{Mussardo} or formula (10) in \cite{KPW})
this quantity (taken with the opposite sign) is called {\it the free energy} per spin for the infinite mean field Ising model. This agrees with our description of the $z$-coordinate in the thermodynamic phase space. It should be mentioned that \cite{FV} uses different terminology.
}
\end{rem}

Next,  we explain how the
Legendrian $\Lambda_\infty$ in (\ref{nov2004}) appears in the thermodynamic limit $N\to+\infty$ from the geometric point of view.
Let $\Lambda_N$ be the equilibrium
Legendrian of the finite Curie-Weiss model with $N$ spins, as given by~(\ref{jul1202}), with $f(q)=-\phi_H(q)$.
Observe that the natural $\R_+$-action on $\mathbb{J}^1\R$: $$c \cdot (p,q,z):= (cp,q,cz)\; \;
 \forall c> 0, (p,q,z) \in \R^3\;$$
 is conformal with respect to the Gibbs contact form
$dz-pdq$. As a result, it maps Legendrians to Legendrians.
Thus,
the rescaled submanifold
%
$$\overline{\Lambda}_N := \frac{1}{N}\Lambda_N= \left\{\left(\frac{p}{N},q,\frac{z}{N}\right)\;:\; (p,q,z) \in \Lambda_N \right\}\;,$$
is also Legendrian. It can be written in the form of (\ref{nov2108})
\begin{equation}\label{eq-barlambdaN}
\overline{\Lambda}_N= \Lambda_{f_N}:= \left\{ z= f_N(q),~
p = \frac{\partial f_N}{\partial q}(q)\right\}\;,
\end{equation}
with $f_N(q) = -N^{-1}\phi_{H_N}(q)$.
Note that by \eqref{eq-ham-red}-\eqref{eq-redCW}, $H_N$ is linear in $q$.  Thus, by Proposition~\ref{prop-secondder},
\be\label{24jul1204}
f_N(q) = \beta^{-1}N^{-1}\ln \mathcal{Z}_N(q)\;,
\ee
where $\mathcal{Z}_N(q)$ is the partition function, and furthermore $f_N(q)$ is smooth and convex.

Recall that {\it the Hausdorff distance} between two compact subsets of $\R^3$ is
the infimum of $\epsilon$ such that each of the subsets lies in the $\epsilon$-neighbourhood
of the other. The next theorem describes the convergence of the Legendrians $\bar\Lambda_N$ to  a part of the Curie-Weiss Legendrian~$\Lambda_\infty$
defined in (\ref{nov2004}) in the thermodynamic limit $N\to+\infty$.

\begin{thm}[Geometric description of the thermodynamic limit] \label{thm-limit}
Write $\pi_q$ for the projection $\R^3(p,q,z) \to \R(q)$.
\begin{itemize}
\item[{(i)}] The sequence of functions $f_N$ defined by (\ref{24jul1204}) converges uniformly to $f$ given by
 \eqref{eq-stablepart} on any compact interval $I\subset \R$ as $N \to +\infty$;
\item[{(ii)}] The derivatives $f'_N$ converge pointwise to $f'$ outside of $q=0$ as $N \to +\infty$.
\item[{(iii)}] The Legendrians $\overline{\Lambda}_N \cap \pi_q^{-1}(I) $ converge in the Hausdorff sense as $N \to +\infty$ to a piece-wise smooth Legendrian $L$ consisting of the {\it stable} part of the Curie-Weiss Legendrian submanifold $\Lambda_\infty$, defined in (\ref{nov2004}),
    restricted to $\pi_q^{-1}(I)$  and a linear segment connecting the points $(p=1,q=0,z=0)$ and $(p=-1,q=0,z=0)$.
\end{itemize}
\end{thm}
We refer to Fig.~4 and Fig.~5 for an illustration of the theorem.

\begin{figure}[!tbp]
  \centering
  \begin{minipage}[b]{0.45\textwidth}
    \includegraphics[width=0.7\textwidth]{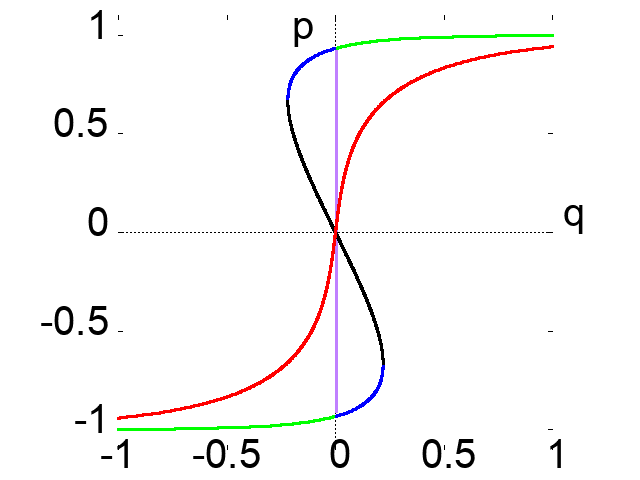}
    \caption{Thermodynamic limit: Lagrangian projection}
  \end{minipage}
  \hfill
  \begin{minipage}[b]{0.45\textwidth}
  \includegraphics[width=0.7\textwidth]{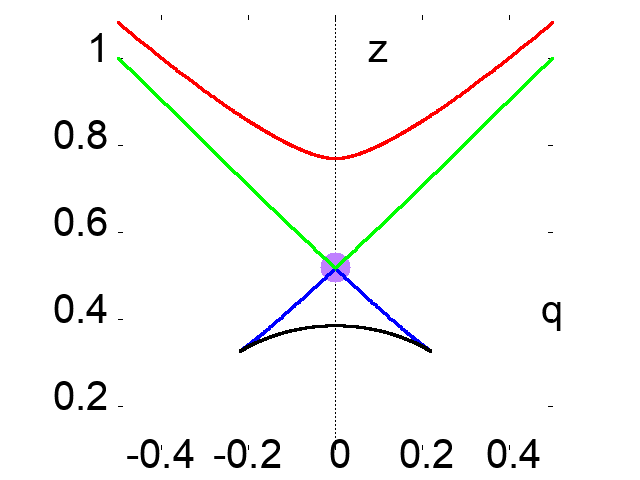}
    \caption{Thermodynamic limit: front projection}
  \end{minipage}
\end{figure}

\begin{proof}
By \cite[Theorem 2.8]{FV}, for each $m$,
$$\lim_{N \to +\infty}  f_N (m) = f(m)\;,$$
where $f$ is given by \eqref{eq-stablepart}.
Since $f_N \to f$ pointwise, and all the functions are convex, the convergence
is uniform on compact subsets \cite[Theorem 7.17]{RW}. This proves item (i).

Observe now that the graph of the {subdifferential} $\partial f$ is the projection of the piece-wise smooth Legendrian $L$ from item {(iii)}
to the $(p,q)$-plane. We refer to \cite[Chapter 8]{RW} for the notion of subdifferential (in \cite{RW} it is called {subgradient}).

Since $f_N\to f$ uniformly on compact sets and since all the functions are
convex and proper, by a theorem of Attouch, the graphs of the
subdifferentials $\partial f_N$, restricted to compact intervals of $\R$,
converge, as sets, to the graph of $\partial f$ in the Hausdorff sense. We
refer to Theorem 12.35 in \cite{RW} combined with some auxiliary material
related to various notions of convergence of set-valued mappings, namely
Example 4.13, Definition 5.32, and Proposition 7.15 {\it ibid.}
Since for a smooth
function $h$ the subdifferential is the usual differential, $\partial h = dh$,
and since $f_N$ are smooth everywhere and $f$ is smooth outside $0$, we get items
(ii) and (iii).
This completes the proof.
\end{proof}


\begin{rem} \label{rem-firstmetas}
 {\rm From the viewpoint of the equilibrium thermodynamics, only the stable part
of $\Lambda_\infty$
has a physical meaning.
However, non-equilibrium thermodynamics as manifested by the
Glauber dynamics highlights the importance of the metastable part, see
Section \ref{subsec-meta} below.
}
\end{rem}

\begin{rem}\label{rem-conv} {\rm Denote by $\pi_{pq}: \R^3 \to \R^2(p,q)$ the projection
of the thermodynamic phase space to the symplectic plane $(p,q)$. Then $\{\pi_{pq}(\overline{\Lambda}_N)\}$ is a sequence of Lagrangian submanifolds which is expected to be Cauchy with respect to (a version of) the $\gamma$-distance \cite{Vi}.
The ``singular (non-smooth) Lagrangian" $\pi_{pq} (L)$
is the $\gamma$-support of the corresponding element in the completion of the space of Lagrangian submanifolds. Let us also comment that the contact geometry of singular Lagrangians and Legendrians has been the subject of active study recently \cite{Vi, Sto, DRS}. The simple example of the thermodynamic limit of the Curie-Weiss magnet provides a setting when such singularities arise naturally. It is reasonable to believe that, more generally, thermodynamic limits of systems with metastable
equilibria lead to such singular Legendrians. Note also that the singularities disappear when the analytic continuation of the singular limiting
Legendrian is performed, as in adding the metastable blue part and the black part depicted in  Fig.~5 to the stable green part of $\Lambda_\infty$. It would be interesting to understand if the
singular Legendrians arising in the thermodynamic limits of other systems share this property
and how their analytically continued part relates to the metastable states.}
\end{rem}

\section{Nonequilibrium {evolution}}\label{sec:noneq}

\subsection{{Gradient flow in the ghost variables}}\label{subsec-noneq}

There exist several mathematical approaches to processes of the non-equilibrium thermodynamics. One of them is based on contact Hamiltonian dynamics where the relaxation processes
are modeled by asymptotic trajectories of contact Hamiltonian systems
(see e.g. \cite{BLN,EP,Goto-JMP2015,Goto,GLP,Grmela,Haslach,LimOh,V}).
In the present paper, {we view Gibbs equilibria as points on the Legendrian
submanifold for which the free energy serves as the generating function. In that context,
a natural relaxation to a Gibbs equilibrium is provided by
another} approach based on the gradient dynamics
associated to the generating functions of Legendrian submanifolds. We
refer to Section \ref{sec-cont} below for a comparison to the contact geometric approach.
{This gradient dynamics is as follows.}
 Let
$\Psi: \R^n \times E \to \R$ be a generating function of a Legendrian
submanifold $\Lambda=\Lambda_\Psi$ given by \eqref{eq-gf}. Fix a Riemannian
metric on  the space $E$ of ghost variables, and for a fixed $q\in\Rm^n$
consider the gradient flow equation
\begin{equation}
\label{eq-grad}
\dot{\xi} = \nabla \Psi(q,\xi).
\end{equation}
If $\Psi$ is proper and bounded from above, $\xi(t)$ necessarily converges to a critical point
of $\Psi$. Thus the trajectory
\begin{equation}
\label{eq-grad-dynam}
\Big(q, p(t):= \frac{\partial \Psi}{\partial q}(q,\xi(t)), z(t)= \Psi(q,\xi(t)\Big)
\end{equation}
converges to a point  on the Legendrian sub-manifold $\Lambda_\Psi$ as $t \to +\infty$.
{In the thermodynamic context, this limiting point is the Gibbs distribution to which the system
relaxes.}

In what follows we shall study equation \eqref{eq-grad} in the framework of
Sections \ref{sec:glossary} and \ref{sec:contact}. {That is, we  model
non-equilibrium thermodynamics as
such relaxation process converging to a point on the Legendrian submanifold of thermodynamic equilibria, as in~(\ref{nov2106}) or
(\ref{nov2108}).} As before, in the thermodynamic setting
the space $E$ of ghost variables
is the space $\mathcal{P}$ of positive probability densities, and the generating
function $\Psi(q,\rho)$ is given by the minus free energy (see \eqref{eq-Phi-rho}):
\begin{equation}\label{eq-PSi-rho}
\Psi: \R^n \times \mathcal{P} \to \R, \;\; (q,\rho) \mapsto \beta^{-1}S(\rho) - \int_M H(q,m) \rho(m)d\mu(m)\;.
\end{equation}
We recall that $\Psi$ is concave in $\rho$, and that we assume that the Hamiltonian $H$ is
linear in~$q$. For a fixed $q$, the maximum of $\Psi$ is attained at the Gibbs distribution
$\rho_G(q)$, and the function $\psi_H(q):= \Psi(q,\rho_G(q))$ is convex in $q$.

{One physical assumption we make is that when
the physical parameters jump from $q_0$ to $q_1$, the generating function (minus free energy)
changes from $\Psi(q_0,\rho)$ to $\Psi(q_1,\rho)$ but the density of states $\rho$ does not. Thus, the initial condition $\rho(0)=\rho_0$ for
the evolution equation~(\ref{eq-grad})
\be\label{24jul1206}
\dot\rho=\nabla\Psi(q_1,\rho),
\ee
is the Gibbs distribution for $\Psi(q_0,\rho)$: $\nabla\Psi(q_0,\rho_0)=0$.}

{Another point is that  one needs to specify the metric in which gradient is taken
in (\ref{24jul1206}). A natural class of metrics on probability densities are
the Wasserstein distances. For the moment we leave the specific choice of the metric open
but it will become concrete when we discuss the Curie-Weiss magnet below.
One should
keep in mind that the specific choice of the metric is also a modeling assumption.}

{Once the metric is specified}, we associate a curve $\gamma(t) = (p(t),q_0, z(t))$
in the thermodynamic phase
space to the solution $\rho(t)$ to (\ref{24jul1206}), where
\begin{equation}\label{eq-rel-fs}
p(t) = \frac{\partial \Psi}{\partial q} (q_0,\rho(t)), \;\; z(t) = \Psi(q_0,\rho(t)).
\end{equation}
Since $\rho(t)$ converges to the Gibbs distribution $\rho_G(q_0)$ as $t \to +\infty$,
the curve $\gamma(t)$ converges to a point $(p,q_0,z)$ on the  Legendrian submanifold (cf. \eqref{nov2108})
\begin{equation}\label{eq-lambda-second}
\Lambda = \Big\{ z= \psi_H(q),~ p = \frac{\partial \psi_H(q)}{\partial q}(q)\Big\}\;.
\end{equation}
This general strategy is related to the version of the GENERIC approach in~\cite{GrO}. In what follows we discuss the thermodynamic limit of such gradient dynamics in the contact geometric context.


\begin{rem}\label{rem-Floer} {\rm Certain interesting classes of Legendrian
submanifolds admit a canonical generating function given by the action functional
on the infinite dimensional space of paths in the phase space. The paths play the role
of the ``ghost variable" $\xi$. The gradient flow of the action functional gives rise
to a version of the Cauchy-Riemann equation which plays the central role in the Floer theory.
We refer to \cite{Oh} where the Lagrangian case is elaborated. It would be extremely interesting to explore a connection between the Cauchy-Riemann and the Fokker-Planck equation
\eqref{eq-FP-discr} appearing below for the finite Curie-Weiss model.}
\end{rem}
In the next sections, we discuss in detail the gradient flow of $\Psi(q,\cdot)$ with respect to
the Wasserstein metric on the space $\cP$ in the example of the Curie-Weiss model and the corresponding thermodynamic limit.

\subsection{{The basin of attraction}}\label{sec:basin}

{Before discussing the time-dependent dynamics,}
we define the basin of attraction $\mathcal{B}$ of a Legendrian $\Lambda$ as the set of all
initial conditions $\gamma(0){=(p_0,q_0,z)}$ that can be realized by a measure $\rho_0\in\cP$ {in (\ref{eq-rel-fs})}, that is,
\begin{equation}\label{eq-basin}
\mathcal{B} = \bigg\{ (p_0,q_0,z_0)
 \in \J^1\R^n \; \Big{|}   \; \exists \rho_0 \in \mathcal{P} \;:\; p_0= \frac{\partial \Psi}{\partial q}(q_0,\rho_0),  \; z= \Psi(q_0,\rho_0)\bigg\}.
 \end{equation}

\begin{rem}\label{rem-observable} {\rm In our setting, the thermodynamic phase space $\J^1\R^n$ is shared by a variety of thermodynamic systems with the same collection of external variables $q \in \R^n$, but with different Hamiltonians and temperatures. The points of the thermodynamic
phase space correspond to macroscopic states of a system for different probability densities on the space of microstates.
With this terminology, one can interpret the basin of attraction of a system as the subset of the thermodynamic phase space which consists of its macroscopic states.}
\end{rem}

The next result provides a constraint on the basin of attraction.
We tacitly assume that the function $\psi_H(q)$ is strictly convex,
and write $\widehat{\psi}$ for the Legendre transform of a function~$\psi$:
\begin{equation}\label{eq-Legendre}
\widehat{\psi} (p) = pq_*-\psi(q_*)\;,\;\;\text{where}\;\; \frac{\partial \psi}{\partial q} (q_*) = p.
\end{equation}

 \begin{prop}
 \label{prop-basin}
 \begin{equation}\label{eq-basin-3}
 \mathcal{B} \subset \bigg\{ (p,q,z)
 \in \J^1\R^n \; \Big{|}   \; z \leq pq - \widehat{\psi_H}(p)\bigg\}\;.
\end{equation}
\end{prop}

Before proving the proposition, we need the following observation.
Put $$\omega(m) = -\frac{\partial H}{\partial q}(q,m),$$
and note that $\omega$ does not depend on $q$ since $H$ is linear in $q$.
Thus, if $p_0=\partial\Psi/\partial q (q_0,\rho_0)$, then
\begin{equation}\label{eq-p0}
p_0= \frac{\partial\Psi}{\partial q} (q_0,\rho_0) = -\int_M \frac{\partial H}{\partial q} (q_0, m) \rho_0 (m) d\mu (m) = \int_M \omega(m)\rho_0 (m) d\mu (m).
\end{equation}
Write $\mathcal{P}(p)$ for the set of all $\rho_{{0}} \in \mathcal{P}$
satisfying \eqref{eq-p0} with $p_0=p$. Then
\begin{equation}\label{eq-basin-1}
\mathcal{B} \subset \bigg\{ (p,q,z)
 \in \J^1\R^n  \; \Big{|}   \;  z \leq \max_{\rho \in  \mathcal{P}(p)} \Psi(q,\rho)\bigg\}\;.
 \end{equation}
Indeed, the coordinate $z$ stands for minus free energy, which is equal
 to $-\Phi(q,\rho) = \Psi(q,\rho)$.

\medskip
\noindent{\it Proof of Proposition \ref{prop-basin}:}
Fix $p,q$  and choose  $q_*$ as in  \eqref{eq-Legendre} with $\psi=\psi_H$.
The Gibbs distribution $\rho_*:= \rho_G(q_*)$ satisfies \eqref{eq-p0} with $q_0=q_*$
because
$$p = \frac{d\psi_H}{dq} (q_*) = \frac{\partial \Psi}{\partial q} (q_*, \rho_G (q_*)),$$
since $\partial\Psi/\partial \rho (q_*,\rho_G (q_*)) =0$.
Furthermore,
$$\Psi(q_*,\rho) \leq  \Psi(q_*,\rho_*) = \psi_H(q_*),\;\;\forall \rho \in \mathcal{P}\;,$$
because $\rho_*$ is the maximum of $\Psi (q_*,\cdot)$.
Since $H$ is linear in $q$ and, accordingly, so is the only term in $\Psi$ depending on $q$ (see \eqref{eq-PSi-rho}),
$$\max_{\rho \in  \mathcal{P}(p)} \Psi(q,\rho) = \max_{\rho \in  \mathcal{P}(p)} \Psi(q_*,\rho) + (q-q_*)p=  \psi_H(q_*) +  (q-q_*)p =$$
$$= p q_* - \widehat{\psi_H} (p) + (q - q_*) p = p q - \widehat{\psi_H} (p).$$
Together with \eqref{eq-basin-1}, this yields the proposition.
\qed

Note that the equilibrium Legendrian $\Lambda$ lies in $\mathcal{B}$, and it also
lies in the hypersurface
\begin{equation}\label{eq-hyp}
\mathcal{H} := \bigg\{ (p,q,z)
\in \J^1\R^n  \; \Big{|}   \;  z = pq -
 \widehat{\psi_H}
 (p)\bigg\}\;.
\end{equation}
Thus $\Lambda \subset \partial \mathcal{B}$, i.e.,  the equilibrium submanifold
lies on the boundary of the basin of attraction.

\subsection{Relaxation of the Curie-Weiss magnet, and contact chords} \color{black}
\label{subsec-chords}
A typical example one may have mind for the discussion above is as follows.
Assume that the temperature of the Curie-Weiss magnet, with the interaction parameter $b>0$
is instantly changed from $T_0 =1/\beta_0$ to $T_1 = 1/\beta_1$, and simultaneously
a source of constant magnetic field $a > 0$ is switched on, so that the external magnetic field $q$ changes to $q+a$. Suppose that before the change the magnet was in a
stable equilibrium state given by the Gibbs distribution  associated to
$b,\beta_0,q$. The issue is how one may describe the relaxation of the magnet towards the
Gibbs distribution associated to
$b,\beta_1, q+a$, and how that relaxation behaves in the thermodynamic limit $N\to+\infty$.
Assume that $b\beta_i > 1$ for $i\in\{0,1\}$, and denote
  by~$\rho_0(q)$ (resp., $\rho_1(q)$) and by $\mathcal{F}_0(q,\rho)$ (resp., $\mathcal{F}_1(q,\rho)$)
the Gibbs distribution and minus free energy associated to $b,\beta_0,q$ (resp., $b,\beta_1, q+a$).
\color{black}
In order to solve the relaxation problem  by the methods of this paper,
we apply the gradient flow equation associated with
 $\mathcal{F}_1$ \color{black}  with the initial condition given by the Gibbs distribution   $\rho_0$. \color{black} As we will see, for a finite $N$,
this gradient evolution corresponds to a discretization of a Fokker-Planck equation.
The Fokker-Planck evolution, as $t \to +\infty$, brings $\rho_0(q)$ to $\rho_1(q)$.
Because the Curie-Weiss Hamiltonian is linear in $q$,  geometrically, the instant temperature jump can be described as a perturbation of the stable
part of the equilibrium Legendrian submanifold associated to $\beta_0$ which keeps $(p,q)$-coordinates
and shifts the (minus) equilibrium energy $z = \mathcal{F}_0(q,\rho_0)$
to $z' = \mathcal{F}_1(q,\rho_0)$. Let us emphasize that the perturbed submanifold is
non-Legendrian, and
that it manifestly lies in the basin of attraction $\mathcal{B}(\beta_1)$.
Similar problems, albeit from a different perspective, were considered in \cite{EK,KH}.

Interestingly, in the thermodynamic limit $N \to +\infty$ this description
of the relaxation is closely related to contact geometry. We assume that
prior to passing to the limit we made the rescaling by factor $1/N$ as in Section \ref{subsec-limit}. Then (assuming $q \neq 0$, $q+a \neq 0$),
the Gibbs distributions for the original system in the thermodynamic limit is given by $\rho_0= \delta(p- p_0(q))$
where
\begin{equation}\label{eq-original-sys}
q+p_0(q) = \beta_0^{-1} \text{arctanh} \left(p_0(q)\right),
\end{equation}
and for the terminal system by
$\rho_1= \delta(p- p_1(q))$,
where
\begin{equation}\label{eq-terminal-sys}
a+q+p_1(q) = \beta_1^{-1} \text{arctanh} \left(p_1(q)\right).
\end{equation}

Assume now that for some $q=q_*$,
\begin{equation} \label{eq-chord} p_0(q_*) = p_1(q_*).
\end{equation}
Geometrically this corresponds to an intersection of the Lagrangian projections
(i.e., projections to the $(p,q)$-plane) of the (rescaled) equilibrium Legendrians of the original and of the terminal system.
The latter is the Curie-Weiss Legendrian corresponding to $b,\beta_1$ shifted by $a$ along the $q$-axis. \color{black}
In purely contact terms, such an intersection point
corresponds to the Reeb chord, i.e., a segment parallel to $z$-axis joining the two
Legendrians. The length of the chord corresponds to the jump of the free energy.
Note that such chords between a pair of Legendrian submanifolds are an object of interest in contact topology and dynamics.
In particular, there exist powerful methods enabling one
to detect such chords in a number of interesting situations, see e.g. \cite{EP} and references therein.

Back to our relaxation story, we get that at $q=q_*$ and, once again, this is possible only in the thermodynamic limit,
the Gibbs distributions of both systems coincide.  Thus  the original system jumps
to the equilibrium point of the terminal system, and the gradient dynamics keeps it fixed.
Therefore, {Reeb chords connecting stable equilibria correspond to instant relaxation processes when the system jumps from one endpoint of the chord to the
other one, and stays there forever.
Note that without the stability assumption we cannot approximate the case of finite but large $N$ by the thermodynamic limit.

It remains to check  the existence of such a chord, that is, of a solution of \eqref{eq-chord}.
Assume that $T_1 > T_0$ and
that $a$ is large enough.
Write $p_0(q)=p_1(q) := p$, $T_i = \beta_i^{-1}$,  and subtract \eqref{eq-original-sys} from  \eqref{eq-terminal-sys}. We readily get
\begin{equation} \label{eq-p-chord}
p = \tanh \frac{a}{T_1-T_0}\;,
\end{equation}
and thus $q_*$ is uniquely defined. Now turn to the analysis of the stability.
Since $a>0$, $T_1 > T_0$ we have $p >0$.
Denote by $c_i$ the positive solution of the equation
$$bx= T_i \text{arctanh}(x)\;.$$
Then $1 > x_1 > x_0$, and the desired stability assumption would follow from
\begin{equation} \label{eq-p1-chord}
p = \tanh \frac{a}{T_1-T_0} > x_1\;.
\end{equation}
This holds true if $a$ is sufficiently large, and the claim follows.

\subsection{The Curie-Weiss basin of attraction}\label{subsec-basin}

In this section, we describe (up to a small error tending to $0$ as $N\to +\infty$) the
basin of attraction $\mathcal{B}_N$ of the Legendrian submanifold
$\overline{\Lambda}_N$ of the Curie-Weiss model with $N$ spins, see
\eqref{eq-barlambdaN}.
We write $o(1), O(1)$ with respect to $N\to +\infty$ below.

We start with the following preparations.  Recall from
Sections~\ref{sec:glossary}, \ref{finite-CW} that the minus free energy of
the Curie-Weiss model with $N$ spins, after the rescaling by factor $1/N$, is
given by
\begin{equation}\label{eq-psiN}
\Psi_N(q,\rho) = \Psi_N(0,\rho) + pq = \frac{1}{N}\beta^{-1}S(\rho) - \int_{M_N} F_{b,\beta}(0,m) \rho(m) d\mu(m) + pq + o(1)\;,
\end{equation}
where
\begin{equation}
\label{eq-p-mean-spin}
p = \int_{M_N} m\rho(m)d\mu(m)\;.
\end{equation}

Consider the function $G^+_{b,\beta}(p)$, $p \in [-1,1]$,  which is obtained from
$-F_{b,\beta}(0,p)$ by applying twice the Legendre transform and, at the end,
changing the sign. This is the concave  envelope of $F_{b,\beta}$, i.e. the
smallest concave function majorizing $F_{b,\beta}$. Let  $G^-_{b,\beta}(p)$
be obtained from the function $F_{b,\beta}(0,p)$ by applying twice the
Legendre transform.  This is the convex envelope of $F_{b,\beta}(0,p)$, i.e.
the greatest convex function majorized by $F_{b,\beta}(0,p)$.

\begin{figure}[htb]
\centering
\includegraphics[width=0.6\textwidth]{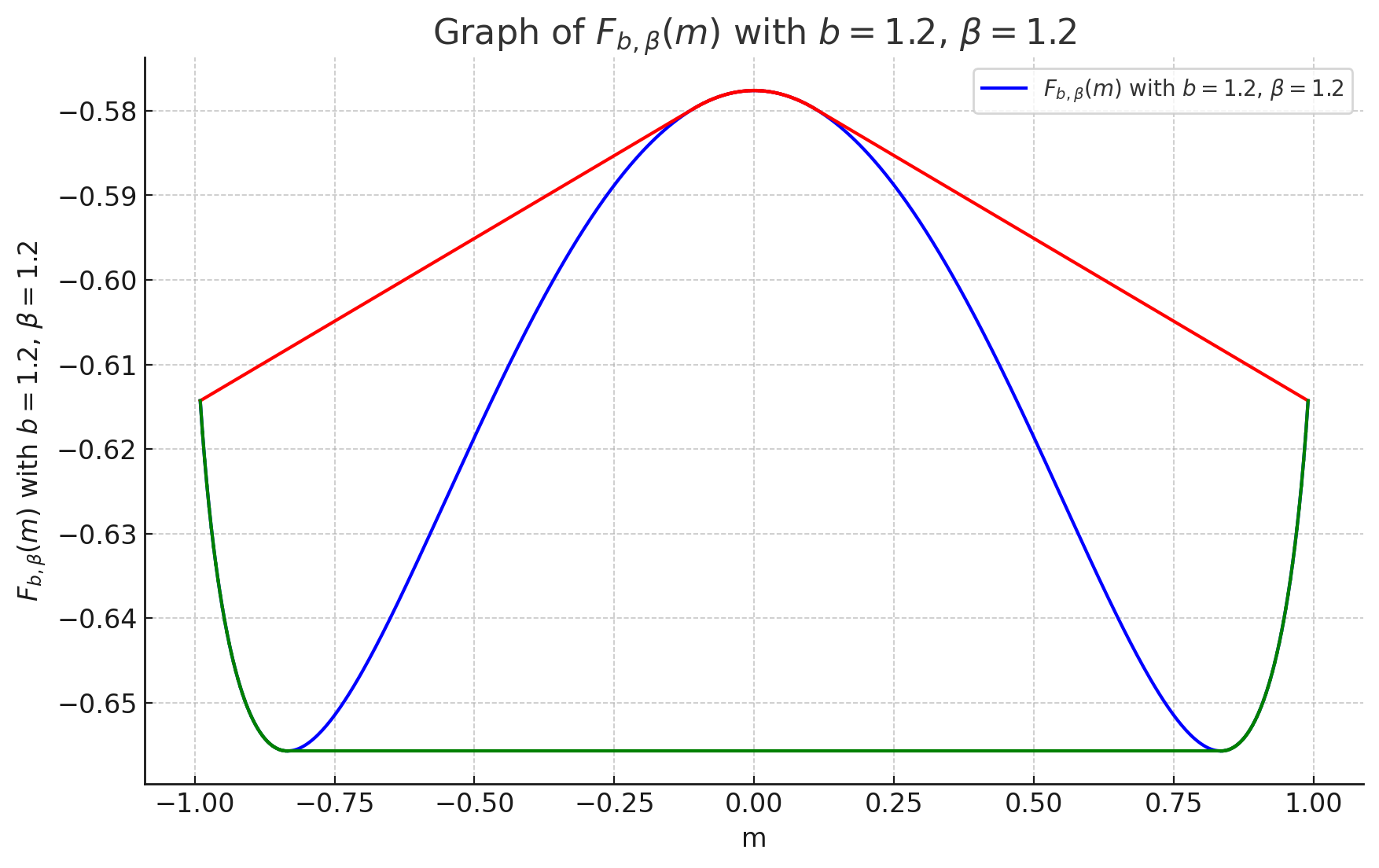}
\caption{Function $F_{b,\beta}$ and its envelopes}
\end{figure}

It is easy to see that when $b\beta \in (1, 2\ln 2)$ the concave envelope
$G^+_{b,\beta}$ is constant: $G^+_{b,\beta}(p) = -b/2$ for all $p$.
If $b\beta > 2\ln 2$, the function $G^+_{b,\beta}$ is more complicated, see Fig.~6,
where it appears in red. The convex envelope $G^-_{b,\beta}$ is plotted in green in the same figure.

Put
$$\alpha_{N,q_0}^{-} (p)  :=  pq_0 - \sup_{\rho \in \mathcal{P}(p)} \Psi_N(q_0,\rho)
= - \sup_{\rho \in \mathcal{P}(p)}\Psi_N(0,\rho)\;,$$and
$$\alpha_{N,q_0}^{+} (p) := 
 pq_0 - \inf_{\rho \in \mathcal{P}(p)} \Psi_N(q_0,\rho) = - \inf_{\rho \in
\mathcal{P}(p)}\Psi_N(0,\rho)\;,$$ where, as in
Section~\ref{subsec-noneq}, $\mathcal{P}(p)$ denotes the set of measures
$\rho\in\mathcal{P}$ satisfying \eqref{eq-p-mean-spin}.
\begin{prop}
\label{prop-basin-descr} The basin of  attraction $\mathcal{B}_N$
is given by
\begin{equation}\label{eq-basin-4}
\mathcal{B}_N =  \bigg\{ (p_0,q_0,z_0)
 \in \J^1\R^3 \; \Big{|}   \; q_0 \in \R, p_0 \in (-1,1),  
\alpha_{N,q_0}^-(p_0) < p_0q_0-z_0 
 \leq  \alpha_{N,q_0}^+(p_0)  \bigg\}\;,
\end{equation}
where
\begin{equation}\label{eq-alphaminus}
\alpha_{N,q_0}^-(p_0) = G^-_{b,\beta}(p_0)+
r^-_{N,q_0} (p_0)\;,
\end{equation}
and
\begin{equation}\label{eq-alphaplus}
\alpha_{N,q_0}^+(p_0) = G^+_{b,\beta}(p_0) +
 r^+_{N,q_0} (p_0)\;,
\end{equation}
where $r^-_{N,q_0} (p_0), r^+_{N,q_0} (p_0)\to 0$, for fixed $p_0,q_0$, as
$N\to +\infty$.
\end{prop}

\begin{proof}
By definition, $p$ is the mean spin with respect to a positive probability measure, hence it always lies in $(-1,1)$. For fixed $p_0\in (-1,1)$ and $q_0\in\R$, the set
$$\mathcal{B}_N (p_0,q_0):= \{ z_0\in\R\ |\ (p_0,q_0,z_0)\in \mathcal{B}_N\} \subset\R$$
is exactly the set of the values of $\Psi_N (q_0,\cdot)$ on
$\mathcal{P} (p_0)$. Note that this set does contain its supremum, reached on the Gibbs {\sl positive} probability measure $\rho_G$ (see Section~\ref{sec:glossary}). At the same time it does {\sl not} contain its infimum, since $\Psi_N (q_0,\cdot)$ is a concave function and $\mathcal{P} (p_0)$ is an open subset of an affine space.
In other words, 
\[
z_0\in \mathcal{B}_N (p_0,q_0)
\Longleftrightarrow
\inf_{\rho \in \mathcal{P}(p_0)} \Psi_N(q_0,\rho) < z_0 \leq \sup_{\rho \in \mathcal{P}(p_0)} \Psi_N(q_0,\rho)\Longleftrightarrow
\]
\[
\Longleftrightarrow 
\alpha_{N,q_0}^-(p_0) < p_0q_0-z_0 \leq \alpha_{N,q_0}^+(p_0).
\] 
This readily yields \eqref{eq-basin-4}. It also yields
\begin{equation}
\label{eq-alphas-extrema}
\alpha_{N,q_0}^-(p_0) = \inf_{z_0\in \mathcal{B}_N (p_0,q_0)} (p_0q_0-z_0),\ \alpha_{N,q_0}^+ (p_0) = \sup_{z_0\in \mathcal{B}_N (p_0,q_0)} (p_0q_0-z_0).
\end{equation}

 Recall that the role of the function 
 $\psi_H$ 
 defined in Section \ref{subsec-noneq} is played 
 in the setting of this section
 by the 
 function 
 $f_N$ from formula \eqref{eq-barlambdaN}.
By using Theorem \ref{thm-limit} and rewriting formula \eqref{eq-stablepart}, we get 
 that, as $N\to +\infty$, the functions $f_N$ converge uniformly on compact sets to
 $$f(q) : =
 \max_m \left(-F_{b,\beta}(q,m)\right) =
 \max_m \left(mq-F_{b,\beta}(0,m)\right)\;,$$ which is the Legendre transform of  $F_{b,\beta}(0,\cdot)$.  
  Thus the Legendre transform $\widehat{f}$ of $f$ satisfies $\widehat{f} =  G^-_{b,\beta}$
and the Legendre transforms of $f_N$ converge to it pointwise as $N\to +\infty$
(see Theorem 11.34 combined with Theorem 7.17 in \cite{RW}).
Together with  \eqref{eq-alphas-extrema} and \eqref{eq-basin-3}, this yields \eqref{eq-alphaminus}.

Now let us prove the estimate for
 $\alpha_{N,q_0}^+(p_0)$.
Observe that $S(\rho) = O(1)$ (the maximum of the entropy is attained on the uniform
distribution with 
 $\rho(m) = N/(2N+2)$ for all $m$). 
 Thus \eqref{eq-alphaminus} yields
\begin{eqnarray}
\label{eq-inf-vsp}
- \sup_{\rho \in \mathcal{P}
(p_0)
}
\int_{M_N} -F_{b,\beta}(0,m) \rho(m) d\mu(m) = 
 - \sup_{\rho \in \mathcal{P}(p_0)}\Psi_N(0,\rho) + o(1) = \nonumber \\
 = \alpha_{N,q_0}^-(p_0)  + o(1) =
G^-_{b,\beta} (p_0) + r^-_{N,q_0} (p_0) +o(1)\;.
\end{eqnarray}
 Similarly, \eqref{eq-alphaplus} would follow from
$$
- \inf_{\rho \in \mathcal{P}
(p_0)
 } 
\int_{M_N} -F_{b,\beta}(0,m) \rho(m) d\mu(m) = G^+_{b,\beta} (p_0) +
 r^+_{N,q_0} (p_0),
$$
 where $r^+_{N,q_0} (p_0)\to 0$, for fixed $p_0,q_0$, as $N\to +\infty$.
Rewrite the latter equality as
\begin{equation}\label{eq-inf-vsp-1} - \sup_{\rho \in \mathcal{P} (p_0)
 }  \int_{M_N} F_{b,\beta}(0,m) \rho(m) d\mu(m) = -G^+_{b,\beta}
 (p_0) + r^+_{N,q_0} (p_0).
\end{equation}
 Formula \eqref{eq-inf-vsp-1} is nothing else but \eqref{eq-inf-vsp}, which we
just proved, applied to $-F_{b,\beta}$,
 with $r^+_{N,q_0} (p_0) = r^-_{N,q_0} (p_0) + o(1)$.
 Furthermore, the concave envelope
of $F_{b,\beta}$ is the convex envelope of $-F_{b,\beta}$ taken with the minus sign.
Thus, \eqref{eq-inf-vsp-1} follows from \eqref{eq-inf-vsp}, and the proposition follows.
\end{proof}

Let us comment on the limiting shape of the Curie-Weiss basin of attraction
in the thermodynamic limit $N \to +\infty$. Recall from \eqref{nov2004} that the Curie-Weiss equilibrium Legendrian (for $N=+\infty$) associated to
parameters $b,\beta$ is given by
\[
\Lambda_{b,\beta}:= \Lambda = \left\{
\frac{\partial F_{b,\beta}}{\partial p} (q,p)=0,~
~z= - F_{b,\beta}(q,p)\right\} \subset \J^1\R^1.
\]
A direct calculation (see e.g. formulas for $q(p)$ and $z(p)$ on p.2 and p.3 in \cite{GLP}) shows that for $(p_0,q_0,z_0)$ lying on
 $\Lambda_{b,\beta}$
 one has
\begin{equation}\label{eq-qofp}
q_0 = -bp_0 + \frac{1}{2\beta} \ln \frac{1+p_0}{1-p_0}\;,
\end{equation}
and
\begin{equation}\label{eq-zofp}
z_0 = \frac{1}{2\beta}\ln\frac{4}{1-p_0^2} - \frac{bp_0^2}{2}    \;,
\end{equation}
which readily yields
\begin{equation}\label{eq-zminpq}
p_0q_0-z_0 = F_{b,\beta}(0,p_0)\;.
\end{equation}
Denote by $W(\beta)$ the convex hull of the graph of the function $w=F_{b,\beta}(0,p)$
on $p \in (-1,1)$ lying in the $(p,w)$-plane.  By Proposition \ref{prop-basin-descr},
the limiting shape of the basins
 of attraction (as $N \to +\infty$)  equals
\begin{equation}\label{eq-limiting}
\mathcal{B}_\infty(\beta) =  \bigg\{ (p_0,q_0,z_0)
 \in \J^1\R^3 \; \Big{|}   \; q_0 \in \R, p_0 \in (-1,1),  p_0q_0-z_0 \in W(\beta)  \bigg\}\;.
\end{equation}

We conclude with an illustration of Proposition \ref{prop-basin-descr}. Fix  the interaction parameter $b$, and assume, for all values of $\beta$ considered below that $b\beta >1$.
Denote by $\Gamma(\beta)$ the stable part of the Curie-Weiss Legendrian submanifold
$\Lambda_{b,\beta}$. The next result shows, in the terminology of Remark \ref{rem-observable}, that the points of $\Gamma(\beta_0)$ appear as macroscopic states of the infinite Curie-Weiss model with the inverse temperature $\beta_1 > b^{-1}$ if and only if $\beta_1 \leq \beta_0$, i.e., its temperature is not lower than $\beta_0^{-1}$.

\begin{thm}\label{thm-phys}
\begin{itemize}
\item[{(i)}] If $\beta_0 > \beta_1$, then $\Gamma(\beta_0) \subset \mathcal{B}_\infty(\beta_1)$.
\item[{(ii)}] If $\beta_0 < \beta_1$, then $\Gamma(\beta_0) \cap \mathcal{B}_\infty(\beta_1) = \emptyset$.
    \end{itemize}
\end{thm}

\begin{figure}[htb]
\centering
\includegraphics[width=0.6\textwidth]{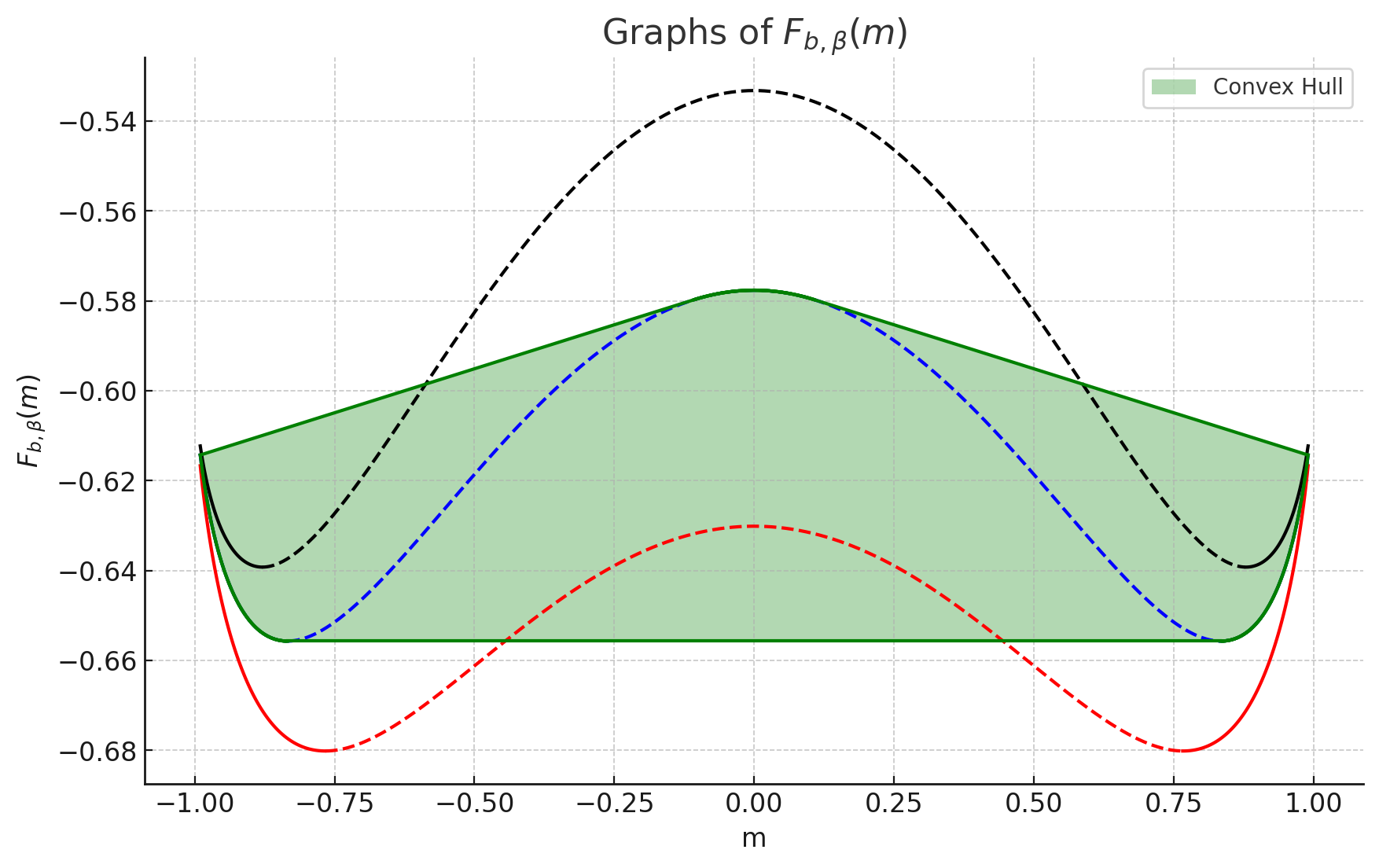}
\caption{Graphs of $F_{b, {\beta_i}}$ for $b=1.2$, $\beta_1 = 1.2$ (blue), $\beta_0 = 1.1$ (red) and $\beta_0= 1.3$ (black). Domain $W(1.2)$ (shaded green) vs. $V(1.3)$ (solid black lines) and $V(1.1)$ (solid red lines).}
\end{figure}

Some comments are in order. Denote by $p_-(\beta)$ and  $p_+(\beta)$ the negative and the positive solutions of the mean-field equation $p= \tanh(b\beta p)$ (cf.\eqref{nov2102}). They correspond
to the points of minimum of the function $F_{b,\beta}(0,p)$.   Take
the Curie-Weiss Legendrian  $\Lambda_{b,\beta}$
 and apply to it the transformation
$$(p,q,z) \mapsto (p,w),\;\; \text{where}\;\; w= pq-z\;.$$
By \eqref{eq-zminpq}, its image coincides with the graph of $w=F_{b,\beta}(0,p)$.
The image of the stable part $\Gamma(\beta)$ under this transformation corresponds
to the restriction of the graph to the intervals $I_-(\beta) = (-1,p_-(\beta))$ and $I_+(\beta) = (p_+(\beta),1)$. Denote this restriction by $V(\beta)$.
Thus, in order to prove Theorem \ref{thm-phys} we have to explore the mutual position of $V(\beta_0)$ and $W(\beta_1)$, where the latter set is defined just before \eqref{eq-limiting}.
We refer to Fig.~7 for an illustration. Let us emphasize that this figure is plotted in the
$(p,w)$-plane.

\medskip
\noindent
{\bf Proof of Theorem \ref{thm-phys}:}
Let us list the following elementary facts, leaving the verification to the reader:
\begin{itemize}
\item [{(1)}] the
 endpoint
$p_-(\beta)$
of the interval $I_-(\beta)$
(resp., the endpoint $p_+(\beta)$) of the interval $I_+(\beta)$) strictly increases (resp., strictly decreases)
as a function of $\beta$;
 \item [{(2)}]  for
 a
 fixed $p$, the function $F_{b,\beta}(0,p)$ strictly increases with $\beta$;
\item [{(3)}] for
 a
 fixed $\beta$, the function $F_{b,\beta}(0,p)$ strictly decreases on $I_-(\beta)$ and strictly increases on $I_+(\beta)$;
\item [{(4)}] $F_{b,\beta}(0,\pm 1) = - b/2$ for all $\beta$ with $b\beta >1$.
    \end{itemize}
It follows that for $\beta_0 > \beta_1$,
 we have
 $V(\beta_0) \subset W(\beta_1)$,
and for $\beta_0 < \beta_1$, $V(\beta_0)$ lies below $W(\beta_1)$, see Fig.~7 for an illustration.
In this figure $\beta_1 = 1.2$ (blue) and $\beta_0$ is taken as $1.1$ (red) and $1.3$ (black).
This yields the statement of the theorem.
\qed

%
%

\subsection{Gradient flow, Wasserstein metric, and Fokker-Planck equation {for a finite Curie-Weiss model}}
\label{subsec-grad-flow-wasserstein-fokker-planck}

We now discuss  the gradient flow equation \eqref{eq-grad} in the context of the finite Curie-Weiss model.
Recall from Section \ref{finite-CW} that the corresponding state space is the  set
$$M_N = \{-1,-1+2/N, \dots, 1-2/N, 1\} \subset [-1,1]$$
consisting of $N+1$ points. The ghost variables lie in the space $\cP_N$ of
the probability densities on $M_N$, see (\ref{nov2112}). In agreement with
the normalization of the function $f_N$ in the
definition~(\ref{eq-barlambdaN}) of the Legendrian manifold $\bar\Lambda_N$, and
anticipating the thermodynamic limit $N\to+\infty$, the generating function
$\Psi$ that appears in (\ref{eq-grad}) will be taken as
\be\label{24jul1502}
\overline{\Phi}_N:=-\frac{1}{N}\Phi.
\ee
Here, the free energy $\Phi$ is given by
\eqref{eq-Phi-rho} with  $H=H_N(q,m)$ as in
(\ref{eq-ham-red})-(\ref{eq-remainder}).

We now recall the definition of the (discrete version of) the Wasserstein distance on the space~$\cP_N$,
as developed by Maas in \cite{Maas} (cf. \cite{J,Lott}). \footnote{We refer to \cite{EFLS}
for a different application of the Maas' version of the Wasserstein distance to
the Curie-Weiss model. We thank Max Fathi for the reference.}
We will denote
$$m_j:= -1 + \frac{2(j-1)}{N},\; j=1,\dots, N+1\;.$$
Fix a smooth positive function
\begin{equation}\label{eq-u-choice}
u:[-1,1] \to (0,1/2),
\end{equation}
and consider a bistochastic tridiagonal symmetric $(N+1) \times (N+1)$ matrix $K_N= (k_{ij})$ with  the off-diagonal entries
\be\label{nov2114}
k_{i,i+1}= k_{i+1,i} = u(m_{i+1}), \;\; \forall i = 1, \dots, N\;,
\ee
and  the diagonal elements
\be\label{nov2116}
\begin{aligned}
&k_{11} = 1- u(m_2)\;,\\
&k_{ii} = 1- u(m_i)-u(m_{i+1}),\;\forall i=2,\dots, N\;,\\
&k_{N+1,N+1}= 1-u(m_{N+1})\;.
\end{aligned}
\ee
The matrix $K_N$ defines a reversible Markov chain with the stationary  probability measure
$$\pi(m_j) = \frac{1}{N+1}\;.$$
For $x,y > 0$ define the logarithmic mean by
$$\ell(x,y) := \frac{x-y}{\ln x - \ln y},\;\;\text{if}\;\; x \neq y\;,$$
and
$$\ell(x,x) := x\;.$$
For a density $\rho \in \mathcal{P}_N$ define the $(N+1)\times (N+1)$ matrix
$\hat{\rho}$ with the entries
$$\hat{\rho}_{ij} = \ell(\rho(m_i),\rho(m_j))\;.$$
We identify the tangent space $T_\rho \mathcal{P}_N$ with the space of vectors
\[
v \in \R^{N+1}_0=\Big\{v \in \R^{N+1}\Big|~\sum_iv_i =0\Big\}.
\]
Define the following operations:
a discrete gradient, taking $\psi \in \R^{N+1}$ to an $(N+1)\times(N+1)$ matrix $\nabla\psi$ with the entries
\begin{equation}\label{eq-disc-grad}
(\nabla \psi)_{ij}:= \frac{N}{2}\cdot(\psi_i-\psi_j)\;,
\end{equation}
and a discrete divergence, taking an anti-symmetric matrix $A = (a_{ij})$ to
a vector $ \Div A$ with the coordinates
\begin{equation}
\label{eq-disc-div}
\left(\Div A\right)_i = -\frac{N}{2}\cdot\sum_j k_{ij}a_{ij}\;.
\end{equation}
The scaling factor $N/2$ in \eqref{eq-disc-grad} and \eqref{eq-disc-div} is needed
in order to pass to the thermodynamic limit $N \to +\infty$ in the next section.

 It is straightforward to see that for each $v \in T_\rho \mathcal{P}_N=\R^{N+1}_0 $ there
is a unique $\phi_v \in \R^{N+1}_0 $
such that
\begin{equation}\label{eq-v-rho} v = - \Div (\hat{\rho}\odot\nabla\phi_v)\;.
\end{equation}
Here, for two square matrices $A$ and $B$ of the same size   we write
$A \odot B = (a_{ij}b_{ij})$ for their Hadamard product.

\begin{rem} \label{rem-Neumann} {\rm In the thermodynamic limit $N \to +\infty$
equation \eqref{eq-v-rho} becomes
\begin{equation}
\label{eq-v-rho-1}
v = - {\rm div}(\rho u^{1/2} \nabla \phi_v)\;,
\end{equation}
where the functions $v,\phi_v$ are defined on the segment $[-1,1]$ and ${\rm div}$ and $\nabla$ are understood with respect to the
Riemannian metric $u^{-1}dm^2$. Here $\rho u^{1/2}$ is the density of the measure $\rho dm$ with respect to the Riemannian measure
$u^{-1/2}dm$.
Formula \eqref{eq-v-rho-1}
has a counterpart on Riemannian manifolds (see formula (8) in  \cite{Otto}, cf. \cite{Lott}).
Furthermore, one can check that in the limit, the solution $\phi_v$ to \eqref{eq-v-rho-1}
should satisfy the Neumann boundary condition.
}
\end{rem}

With the above definitions, the discrete Wasserstein metric on $\mathcal{P}_N$ is given by the scalar product
\be\label{24jul1504}
(v,w)_W := (\nabla \phi_v, \hat{\rho}\odot\nabla \phi_w),
\ee
where $v,w\in T_\rho {\mathcal P}_N$ and
\be\label{24jul1506}
(A,B):=   \frac{1}{N}   \sum_{ij} k_{ij}a_{ij}b_{ij}\;,
\ee
and the lower index $W$ stands for ``Wasserstein".

A direct calculation based on \cite{Maas} (see Appendix in Section \ref{App-2}
for more details) shows that the gradient of the normalized free energy (\ref{24jul1502}) is
\begin{equation}\label{eq-grad-Wass}
\nabla_W \overline{\Phi}_N (\rho) =
\frac{N}{4\beta}(K_N-{\bf 1})\rho + \Div (\hat{\rho}\odot\nabla F_{b,\beta})+
N^{-1}\beta^{-1}\Div (\hat{\rho}\odot\nabla r_N)\;.
\end{equation}
Here, $\nabla_W$ is the gradient with respect to the Wasserstein metric, and
$F_{b,\beta}$ is the vector of the values of the function $F_{b,\beta}$ at
the $N+1$ points of the space $M_N$.  In the limit $N\to+\infty$, the first
term in the right side is responsible for the diffusion, and the second one
for the drift, respectively.
 The last term is a small remainder.

To summarize, the gradient flow of the free energy with respect to the
discrete Wasserstein metric  on $\cP_N$  is given by the discrete Fokker-Planck equation (cf. \cite{Maas})
\begin{equation} \label{eq-FP-discr}
\dot{\rho} =  \frac{N}{4\beta}(K_N-{\bf 1})\rho + \Div (\hat{\rho}\odot\nabla F_{b,\beta})+
N^{-1}\beta^{-1}\Div (\hat{\rho}\odot\nabla r_N)\;.
\end{equation}
The dynamics of (\ref{eq-FP-discr}) depends, of course, on the choice of the
function $u(m)$ that enters the definitions of the matrix $K_N$ and of the divergence operator $\Div$. We will make it
in the next section.

\subsection{Comparison to the Glauber dynamics}\label{sec-comp}

An important {standard} model of relaxation of the Curie-Weiss model, to which the gradient flow~(\ref{eq-FP-discr}) should be compared,
is the Glauber dynamics. 
It is generated by
the following  discrete time random walk on the set $M_N$, with a time step $\tau>0$.
The jumps by $2/N$ to the nearest neighbor on the right and on the left   happen with
the respective probabilities
$\tau N \Pi_+(m)$ and $\tau N \Pi_-(m)$ with
\begin{equation}\label{eq-trans}
\Pi_-(m) = c(1+m)(1-\theta(m-2/N)),\; \Pi_+(m) = c(1-m)(1+\theta(m+2/N)).
\end{equation}
 In addition, at each time step, the particle stays put with the probability
 \[
 1-\tau N \Pi_+(m)-\tau N \Pi_-(m).
 \]
 The parameter $c>0$ in (\ref{eq-trans}) is responsible for the rate of spin flips, and
\be\label{24jul1602}
\theta(m) = \tanh \beta(q+bm).
\ee
The object of interest is the probability  $\rho= \rho(m,t)$ of finding the system at the point $m \in M_N$ at time $t$. Its evolution is described by the following ``master equation":
\begin{equation}\begin{split}
\label{eq-evol-Gl-1}
\frac{\rho(m_i, t+\tau) - \rho(m_i,t)}{\tau}  & =  N\Pi_+(m_{i-1})\rho(m_{i-1},t) - N\left( \Pi_+(m_i)  + \Pi_-(m_i)\right)\rho(m_i,t) \\ & + N\Pi_-(m_{i+1}) \rho(m_{i+1},t)\;,  i=2, \dots, N,
\end{split}
\end{equation}
\begin{equation}
\label{eq-evol-Gl-left}
\frac{\rho(m_1, t+\tau) - \rho(m_1,t)}{\tau} = - N\Pi_+(m_1)\rho(m_1,t) + N\Pi_-(m_{2}) \rho(m_2,t),
\end{equation}
\begin{equation}
\label{eq-evol-Gl-right}
\frac{\rho(m_{N+1}, t+\tau) - \rho(m_{N+1},t)}{\tau} = - N\Pi_1(m_{N+1})\rho(m_{N+1},t) + N\Pi_+(m_{N}) \rho(m_N,t).
\end{equation}

Our objective is to compare this dynamics with the Fokker-Plank evolution given
by~\eqref{eq-FP-discr}. To this end, we make the following choice {of the Wasserstein metric on the space of probability densities on $M_N$}.

\medskip\noindent
{\bf Important choice:} Take the function $u$ entering the definition (\ref{nov2114})-(\ref{nov2116})
of the matrix $K_N$ as
\begin{equation}\label{eq-u-defin}
u(m):= c(1-m\theta(m)),
\end{equation}
with $\theta(m)$ defined in (\ref{24jul1602}).

We introduce a left-stochastic (i.e., entries in each column sum up to $1$) tri-diagonal $(N+1)\times (N+1)$ matrix $P_N = (\kappa_{ij})$
with
\[
\kappa_{i,i-1} = \Pi_+(m_{i-1}), \kappa_{i,i+1} = \Pi_-(m_{i+1})\;, \kappa_{ii} = 1-\Pi_-(m_i) - \Pi_+(m_i),
\]
and set $W_N = P_N-K_N$. With this notation, we pass to the continuous time limit $\tau\to 0$
in~\eqref{eq-evol-Gl-1},\eqref{eq-evol-Gl-left},\eqref{eq-evol-Gl-right},
and readily derive the following equation:
\begin{equation}\label{eq-Glauber-fin}
\frac{\partial \rho}{\partial t} = N(K_N- {\bf 1})\rho +NW_N\rho.
\end{equation}
Note that {\it the diffusion} term $N(K_N-{\bf 1})\rho$
in  \eqref{eq-Glauber-fin} coincides with the one in \eqref{eq-FP-discr} up to a numerical multiple.
{Let us mention that, when applied to a smooth function, the diffusion term in  \eqref{eq-FP-discr} and (\ref{eq-Glauber-fin})
is  formally of the order
 $O(1/N)$
 while the drift terms in these equations are of the order $O(1)$. Nevertheless, the diffusion
terms can not be neglected as the solutions to  \eqref{eq-FP-discr} and (\ref{eq-Glauber-fin})  are expected to have some oscillations
on the scales  $o(1)$ as $N\to+\infty$.}

Let us consider the respective {\it drift} terms in these two equations:
\[
\Div (\hat{\rho}\odot\nabla F_{b,\beta}) +
N^{-1}\beta^{-1}\Div (\hat{\rho}\odot\nabla r_N)
\]
in (\ref{eq-FP-discr}) and $N{W}_N\rho$ in (\ref{eq-Glauber-fin}).  Fix a
smooth function $\rho$ on $[-1,1]$ and consider its restriction to $M_N$:
\begin{equation}\label{eq-restr}
\rho_N = \rho|_{M_N} \in \R^{N+1}.
\end{equation}

\begin{thm}\label{thm-drifts} [Drift correspondence]
\begin{itemize}
\item [{(i)}] For every $m \in (-1,1)$,
\begin{equation}\label{eq-limit-drift-FP}
\Div (\hat{\rho}_N\odot\nabla F_{b,\beta})+
N^{-1}\beta^{-1}\Div (\hat{\rho}\odot\nabla r_N) =\left(uF'_{b,\beta}\rho(m)\right)'+ O({1}/{N}).
\end{equation}
\item [{(ii)}] For every $m \in (-1,1)$,
\begin{equation}\label{eq-limit-drift-Glaub}
NW_N \rho_N (m) =4c\left((m-\theta(m))\rho(m)\right)'+ O({1}/{N}).
\end{equation}
\item [{(iii)}] Consider the vector fields ${\rm Drift}_{FP}(m): = u(m)F'_{b,\beta}(m)$
and ${\rm Drift}_{G}(m) := 4c(m-\theta(m))$
corresponding to the drifts of the  Fokker-Planck evolution and the Glauber evolution, respectively, in the thermodynamic limit $N \to +\infty$. These fields are proportional with the factor being a smooth positive function, and thus, their non-parameterized trajectories
coincide.
\end{itemize}
\end{thm}

The proof is given in Section \ref{sec-appendix}.

\begin{rem}{\rm The vector field ${\rm Drift}_{G}(m)$ is nothing else but the
right hand side of the Glauber-Suzuki-Kubo ODE, see equation (4.3) in  \cite{Suzuki-Kubo}.
Thus, item (ii) of the theorem is a derivation of this ODE. We thank S.~Shlosman who brought
our attention to a possibility of such a derivation in the context of the Curie-Weiss model.}
\end{rem}

\begin{rem}\label{rem-eq-vs-evolution}
{\rm Let us note that Theorem~\ref{thm-drifts})(i), {together with the comment below
(\ref{eq-Glauber-fin}),} does not imply directly
that solutions to the discrete equation  (\ref{eq-FP-discr}) are well-approximated by the first-order continuous problem
\be
\pdr{\rho}{t}=(u(m)\rho(m)F'(m))',~~m\in[-1,1].
\ee
Moreover, one expects that this ``naive" approximation should fail in general. We believe that the ``correct" continuous approximation
to  (\ref{eq-FP-discr}) is the Fokker-Planck equation
\be\label{nov302}
\pdr{\rho}{t}=\farc{d}{N}(u(m)\rho'(m))'+(u(m)\rho(m)F'(m))',~~m\in[-1,1],
\ee
with a suitably chosen constant $d(\beta)>0$. If the drift $F'(m)$ were regular, then (\ref{nov302}) would
be
supplemented by the Robin zero flux boundary conditions
\be\label{nov304}
\farc{d}{N}\rho'(m)+F'(m)\rho(m)=0,~~\hbox{at $m=\pm 1$.}
\ee
In the present case, as can be seen directly from (\ref{eq-redCW}), while $F(m)$ is bounded on $[-1,1]$, the
derivative $F'(m)$ has a logarithmic singularity at $m=\pm 1$.  To remove it, one may consider the function
\be
\zeta(t,m)=\rho(t,m) e^{NF(m)/d}.
\ee
In light of \eqref{nov302} and \eqref{nov304} ,  it satisfies a divergence-form equation
\be\label{24jul3002}
\pdr{\zeta}{t}=\farc{d}{N}e^{NF(m)/d}(u(m)e^{-NF(m)/d}\zeta'(m))',~~m\in[-1,1],
\ee
with  coefficients that are uniformly bounded in $m\in[-1,1]$, and with the zero-flux boundary condition
\be
\zeta'(-1)=\zeta'(1)=0.
\ee
The latter translates into a Robin type boundary condition for $\rho(t,m)$ of the form
\be
\lim_{m\to\pm 1}\Big(\rho(m)e^{NF(m)/d}\Big)'=0.  
\ee

Even though the diffusion term in (\ref{nov302}) is small,  its solutions 
have derivatives that grow in~$N$, so the diffusion term cannot be dropped.
Roughly speaking,  the drift term governs the evolution of $\rho$  on a large (with $N$)
initial time interval, while the diffusion plays a role at large times as
$t\to +\infty$.
To avoid the lengthy technicalities, we leave the proof of the
continuous approximation of the solutions to the
discrete problem{s (\ref{eq-FP-discr}) and} (\ref{eq-Glauber-fin})  to a future work.}
\end{rem}

\subsection{A comment on metastability} \label{subsec-meta}

{Let us now briefly comment on the metastability of the Curie-Weiss model
in the context of the present paper. The metastability can be seen both in the description of the equilibria and the non-equilibrium
dynamics but in somewhat different ways.}

{First, recall that} the front projection of the equilibrium Legendrian submanifold of the
mean-field Ising model corresponds to the green part of the curve represented
on Fig.~3. The blue and black parts (``the lower triangle") so far did not appear
in our story at all. On the other hand, the metastable piece (appearing in blue on Fig.~2 and Fig.~3) can be reconstructed from the stable (green) one via the analytic continuation.
Indeed, the Lagrangian projection of the equilibrium submanifold to the $(p,q)$-plane is given by the equation
$$q = -bp + \beta^{-1} \text{arctanh} (p)$$
with a real analytic right hand side (see Fig.~2). {That is, while the stable part of the Legendrian submanifold $\Lambda_\infty$ discussed
in Section~\ref{subsec-limit} is singular, this manifold admits an analytic continuation and the metastable part is a component of the resulting
smooth manifold.}

Interestingly enough, the blue part of the equilibrium submanifold consisting
of  metastable equilibria can also
be detected either from the non-equilibrium Glauber Markov
or Fokker-Planck dynamics {that we have discussed above.}
Recall that metastability, roughly speaking, means that in the thermodynamic limit, i.e. as the size $N$ of the model increases, the system spends larger and larger time in the metastable region. {This phenomenon is known both for the Glauber dynamics and Fokker-Planck equations with a small diffusion,
of which (\ref{nov302}) is an example when $N\gg 1$.}
We refer to Theorem 13.1 in \cite{BH}
and to~\cite{IS}  for precise formulations for the Glauber and the Fokker-Planck evolutions, respectively. Remarkably, {as far as the
non-equilibrium dynamics is concerned,} the metastable behavior takes place for large, albeit finite, values of $N$, while the analytic continuation argument is applied to the curve arising in the thermodynamic limit $N \to +\infty$.   It would be interesting to
find a conceptual explanation of this phenomenon  (cf. \cite{Penrose}) .

\section{Proof of the drift correspondence}\label{sec-appendix}
Here we prove Theorem \ref{thm-drifts}.

\subsection{Proof of (i)} For the sake of brevity,
set $F=F_{b,\beta}$ and $\rho=\rho_N$, and write $u(m_i)= u_i$,  $F(m_i)= F_i$, $\rho(m_i) = \rho_i$, etc.
Write, for $2\le i\le N$,
\be\label{nov106}
\bal
\Big[\Div (\hat{\rho}\nabla F)\Big]_i&=-\farc{N^2}{4}\sum_{j=1}^{N+1}
k_{ij}\hat\rho_{ij}(F_i-F_j)\\
&=
-\farc{N^2}{4}\Big(k_{i,i-1}\hat\rho_{i,i-1}(F_i-F_{i-1})
+k_{i,i+1}\hat\rho_{i,i+1}(F_i-F_{i+1})\Big)\\
&=-\farc{N^2}{4}\Big(u_i\hat\rho_{i,i-1}(F_i-F_{i-1})
+u_{i+1}\hat\rho_{i,i+1}(F_i-F_{i+1})\Big).
\enbal
\ee
We will expand (\ref{nov106}) up to the order $O(1/N)$.
First, we note that
\be\label{nov112}
\bal
\hat\rho_{i,i+1}&=\frac{\rho_i-\rho_{i+1}}{\ln \rho_i- \ln \rho_{i+1}}=
\frac{\rho_{i+1}-\rho_i}{ \ln\Big(1+\farc{\rho_{i+1}-\rho_i}{\rho_i}\Big)}\\
&=
(\rho_{i+1}-\rho_i)\Big(\farc{\rho_{i+1}-\rho_i}{\rho_i}-\farc{1}{2}\Big(\farc{\rho_{i+1}-\rho_i}{\rho_i}
\Big)^2+O(1/N^2)\Big)^{-1}\\
&=\Big(\farc{1}{\rho_i}-\farc{1}{2}\farc{\rho_{i+1}-\rho_i}{\rho^2_i}+O(1/N^2)\Big)^{-1}=
\Big(\farc{1}{\rho_i}-\farc{\rho'_i}{N\rho^2_i}+O(1/N^2)\Big)^{-1}\\
&=\rho_i\Big(1-\farc{\rho'_i}{N\rho_i}+O(1/N^2)\Big)^{-1}=\rho_i\Big(1+\farc{\rho'_i}{N\rho_i}+O(1/N^2)\Big)
\\
&=\rho_i+\farc{1}{N}\rho'_i+O(1/N^2).
\enbal
\ee
It follows from (\ref{nov112})   that
\be\label{nov114}
\bal
\hat\rho_{i,i-1}&= \rho_{i-1}+\farc{1}{N}\rho'_{i-1} +O(1/N^2)=\rho_i-\farc{2}{N}\rho'_i +\farc{1}{N}\rho'_i+O(1/N^2)\\
&=\rho_i-\farc{1}{N}\rho'_i +O(1/N^2) .
\enbal
\ee
We set
\be\label{nov118}
\bal
\Big[\Div (\hat{\rho}\odot\nabla F)\Big]_i&=-\farc{N^2}{4}\Big(u_i\hat\rho_{i,i-1}(F_i-F_{i-1})
-u_{i+1}\hat\rho_{i,i+1}(F_{i+1}-F_i)\Big)\\
&=-I_1+I_2,
\enbal
\ee
and use (\ref{nov114}) a  to write
\be\label{nov116}
\bal
&
I_1:=\farc{N^2}{4}u_i\hat\rho_{i,i-1}(F_i-F_{i-1})\\
&=\farc{N^2}{4}u_i\Big(\rho_i-\farc{1}{N}\rho'_i
 +O(1/N^2)\Big)
 \Big(\farc{2F'_i}{N}-\farc{2F''_i}{N^2}+
 o(1/N^2)
  \Big)\\
&= \farc{N}{2}u_i\Big(\rho_i-\farc{1}{N}\rho'_i
+O(1/N^2)\Big)
 \Big(F'_i -\farc{F''_i}{N}+O(1/N^2)\Big)\\
&=\farc{N}{2}\Big[u_i\rho_iF'_i-\farc{1}{N}u_i\rho'_iF'_i
-\farc{u_i\rho_iF''_i}{N}+O(1/N^2)\Big].
\enbal
\ee
The second term in (\ref{nov118}) can be written as
\be\label{nov120}
\bal
&I_2:=\farc{N^2}{4}u_{i+1}\hat\rho_{i,i+1}(F_{i+1}-F_i)=
\farc{N^2}{4}\big(u_i+\farc{2}{N}u'_i+O(1/N^2)\big)\\
&\times\Big(\rho_i+\farc{1}{N}\rho'_i +O(1/N^2)\Big)
\Big(\farc{2F'_i}{N}+\farc{2F''_i}{N^2}+
 o(1/N^2)
 \Big)\\
&=\farc{N}{2}\big(u_i+\farc{2}{N}u'_i+O(1/N^2)\big)\Big(\rho_i+\farc{1}{N}\rho'_i+O(1/N^2)\Big)
\Big(F'_i +\farc{F''_i}{N}+O(1/N^2)\Big)\\
&=\farc{N}{2}\Big(u_i\rho_iF'_i+\farc{1}{N}u_i\rho'_iF'_i
+\farc{u_i\rho_iF''_i}{N}+\farc{2}{N}u'_i\rho_iF'_i+O(1/N^2)\Big).
\enbal
\ee
Combining (\ref{nov116}) and (\ref{nov120}) gives
\be\label{oct2718}
\bal
&\Big[\Div (\hat{\rho}\odot\nabla F)\Big]_i=-\farc{N}{2}\Big[u_i\rho_iF'_i-\farc{1}{N}u_i\rho'_iF'_i
-\farc{u_i\rho_iF''_i}{N} \\
&-u_i\rho_iF'_i-\farc{1}{N}u_i\rho'_iF'_i
-\farc{u_i\rho_iF''_i}{N}\\
&-\farc{2}{N}u'_i\rho_iF'_i+O(1/N^2)\Big]\\
&=-\farc{N}{2}\Big[-\farc{2}{N}u_i\rho'_iF'_i-\farc{2u_i\rho_iF''_i}{N}
-\farc{2}{N}u'_i\rho_iF'_i+O(1/N^2)\Big]\\
&=(u\rho F')' (m_i)+O(1/N).
\enbal
\ee

Finally, let us deal with the remainder.
Rewrite $r_N$, by using \eqref{eq-remainder}, in the form
$$r_N(m) = g_0(N) + g_1(m) + N^{-1}g_2(m) +
 O(1/N^2)
 \;,$$
where $g_0(N)$ is a constant, and $g_1,g_2$ are smooth functions of $m$.
Arguing as above we see that for $l=1,2$
$$\Div (\hat{\rho}\odot\nabla g_l) = (u\rho g_l')' +O(1/N)\;.$$
Furthermore, it follows from the definitions that
$$\Div (\hat{\rho}\odot\nabla
 O(1/N^2))
 = O(1)\;.$$
Thus,
$$N^{-1}\beta^{-1}\Div (\hat{\rho}\odot\nabla r_N)=
 O(1/N)
\;.$$
Together with \eqref{oct2718},
this completes the proof.

 \qed

\subsection{Proof of (ii)} Put $a(m) = c(m - \theta(m))$.
Set $\delta:=2/N$, considering it as a small parameter. For the sake of brevity we write $u(m_i)= u_i$,  $a(m_i)= a_i$, $\theta(m_i) = \theta_i$. Denote $$u^*_i = (u_i-u_{i-1})/\delta\;.$$
We use that
\begin{equation}\label{eq-rel-a}
\frac{1}{2} \left(\Pi_-(m_{i+1}) + \Pi_+(m_{i-1})\right) = u_i + c\delta\;,
\end{equation}
and
\begin{equation}\label{eq-rel-u}
\frac{1}{2} \left(\Pi_-(m_{i+1}) - \Pi_+(m_{i-1})\right) = a_i- c \theta_i\delta \;.
\end{equation}
Writing $W_N =(w_{ij})$ and calculating, we get
$$w_{i,i-1} = - a_i + c\delta(\theta_i+1),\; w_{i,i+1} = a_i +c\delta (-\theta_i - u^*_{i+1}+1)\;.$$
Furthermore,
$$w_{ii}= - (w_{i,i-1} + w_{i+1,i})\;.$$
Thus
$$W_N\rho_N (m_i) = a_i(\rho_{i+1} - \rho_{i-1}) + \rho_i (a_{i+1}-a_{i-1}) + c\delta (R_1 - R_2-R_3)\;,$$
where
$$R_1 = \rho_{i-1} - 2\rho_i + \rho_{i+1} = O(\delta^2)\;,$$
$$R_2 = \rho_i (\theta_{i+1}-\theta_{i-1}) + \theta_i(\rho_{i+1}-\rho_{i-1})= O(\delta)\;,$$
and
$$R_3 = u_{i+1}^*\rho_{i+1} - u_{i}^* \rho_i
= \delta^{-1}\left((u_{i+1}-u_i)\rho_{i+1} - (u_i - u_{i-1})\rho_i\right)$$
$$ = \delta^{-1}\left( u'_i\rho_i\delta - u'_i\rho_i\delta +  O(\delta^2)\right) = O(\delta)\;.$$
 It follows that
$$W_N\rho_N (m_i) = (a\rho)' (m_i) \cdot (2\delta) + O(\delta^2)\;.$$
Thus
$$NW_N\rho_N(m_i) = (4a\rho)' + O(1/N)\;,$$
as required. \qed
\begin{figure} [htb]
\centering
\includegraphics[width=0.6\textwidth]{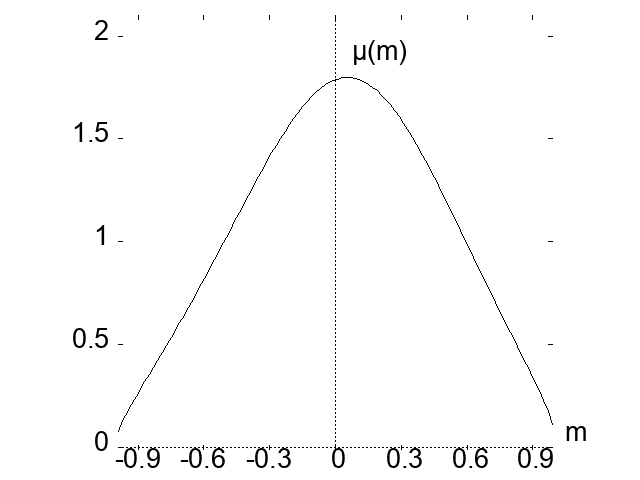}
\caption{Function $\mu$ for $q=-0.08,b=1,\beta=1.8$}
\end{figure}
\subsection{Proof of (iii)} One readily checks that
the functions $m-\theta(m)$ and $$F'_{b,\beta}(m) = -q-bm+\beta^{-1}\text{arctanh}(m)\;$$
have the same zeroes, and furthermore, their ratio
\begin{equation}\label{eq-mu}
\mu(m) = \frac{m-\theta(m)}{-q-bm+\beta^{-1}\text{arctanh}(m)}\;,
\end{equation}
is a smooth positive function, see Fig.~8.
It follows that
\begin{equation}\label{eq-drifts}
{\rm Drift}_{G}= 4u\mu {\rm Drift}_{FP}\;,
\end{equation}
as requred. \qed

\section{Discussion: Fokker-Planck vs. contact dynamics}\label{sec-cont}

In this section, we outline a connection between the Fokker-Planck, {or gradient flow,} dynamics and
contact dynamics, in the following toy example.
Let $M \subset \R^n$ be a strictly convex compact domain,
$h: M \to \R$ be a strictly convex function, and $U$ be the image
of~$M$ under the map $m \to h'(m)$, where prime here and below stands
for $\partial h /\partial m$.
We write $d\mu$ for the Euclidean volume form.

Consider a thermodynamic system whose space of microscopic states is $M$,
with the  Hamiltonian
$H_N = NH$, where
\[
H(q,m):=-qm+h(m),~m \in M,~q \in U,
\]
and $N{\gg 1}$
is a large parameter. The Gibbs distribution
\[
\rho_N = \mathcal{Z}_N^{-1}e^{-\beta H_N},
\]
for $N \gg 1$,
can be approximated by
\be\label{24jul1606}
\rho_\infty={\delta(m-m_*(q))}.
\ee
Here, $m_*(q):= (h')^{-1}(q)$, is the point where $H(q,m)$ attains its minimum in $m\in M$.
Let us also denote by $h^*: U \to \R$ the Legendre transform of $h$, {given explicitly by
\be
h^*(q)=qm^*(q)-h(m^*(q)),
\ee
so that
\be\label{24jul1608}
(h^*)'(q)=m^*(q).
\ee
}

The rescaled pressure and minus free energy are given by
\begin{equation}\label{eq-cont-1}
p = -N^{-1}\frac{\partial \Phi}{\partial q} (q,\rho)= \int_M m\rho(m)d\mu\;,
\end{equation}
and
\begin{equation}\label{eq-cont-2}
z = -N^{-1}\Phi(q,\rho) = (\beta N)^{-1}S(\rho) - \int_M H(q,m)\rho(m)d\mu\;.
\end{equation}
Substituting the approximation $\rho=\rho_\infty$ into \eqref{eq-cont-1},\eqref{eq-cont-2} and taking the limit $N \to +\infty$ {that
eliminates the entropy term in (\ref{eq-cont-2}),} we get
a relation between $p,q,z$ that can be expressed as a
Legendrian submanifold
\begin{equation}\label{eq-cont-lambdainf}
\Lambda_\infty = \left\{ p= (h^*)'(q),\; z = h^*(q) \right\} \subset  \J^1 \R^n\;.
\end{equation}
{We used above (\ref{24jul1606}) and (\ref{24jul1608}).}
The submanifold $\Lambda_\infty$
can be viewed as a limit, as $N \to +\infty$,
of the equilibrium Legendrians
corresponding to the parameter value $N$. Note that it is a counterpart, in the context of our toy model, of the Curie-Weiss
Legendrian \eqref{nov2004}.

Similar to the discrete Fokker-Planck equation deduced in Section~\ref{subsec-grad-flow-wasserstein-fokker-planck},
we can derive the continuous Fokker-Planck equation for $\rho$
with respect to the Euclidean metric on $M$, {as the gradient flow dynamics of the rescaled free energy defined
in (\ref{eq-cont-2}), with respect to the  Wasserstein distance on the space of probability measures $\cP(M)$}.
The function $\rho$ in the equation should satisfy a Robin boundary condition similar to \eqref{nov304}.
The diffusion part
 in the equation {comes from the entropy term in (\ref{eq-cont-2}) and}
 is of the order $N^{-1}$.
 Let us discard it (this is a non-rigorous step!)
 and remain with the drift
\be\label{24jul1610}
\dot{\rho} = \text{div} (\rho H').
\ee
Let us restrict  the Fokker-Planck equation {(\ref{24jul1610})}
to the
delta-functions $\rho{(t,m)}={\delta(m-p(t))}$
concentrated at the interior points of $M$.
In this case, {(\ref{24jul1610})} becomes
the minus gradient equation
(with respect to
the $p$ variable!):
 \begin{equation}\label{eq-minusgr-1}
\dot{p} = - H'(q,p) = q- h'(p)\;.
\end{equation}
By \eqref{eq-cont-2}, {the corresponding free energy is}
\begin{equation}\label{eqref-z-vsp}
z = -H(q,p) = qp-h(p)\;.
\end{equation}
Thus, as $t \to +\infty$,
 the
 trajectories
$(p(t),q(t),z(t))$ converge to the point $((h')^{-1}(q),q, h^*(q))$.
Since $ (h')^{-1}(q)= (h^*)'(q)$ {by (\ref{24jul1608}),} we see that this limiting point lies
on the Legendrian submanifold $\Lambda_\infty$.
Until now, we simply summarized the construction of the Fokker-Planck drift in a toy example.
Our next goal is to relate it to contact dynamics.

To this end let us consider an ad-hoc contact Hamiltonian
\[
E(p,q,z)=-z+h^*(q).
\]
The corresponding contact dynamics is given by
\be\label{24jul1612}
\bal
&\dot{p}=\frac{\partial E}{\partial q} + p\frac{\partial E}{\partial z}=    -p + (h')^{-1}(q),\\
&\dot{q}=-\frac{\partial E}{\partial p} =0,\\
&\dot{z} = E -p\frac{\partial E}{\partial p} = -z + h^*(q).
\enbal
\ee
It readily follows that $\Lambda_\infty$ is an attractor for the contact Hamiltonian flow.
Since $q$ is constant in both Fokker-Planck and contact dynamics, we focus on the evolution
of the $p$-variable. Consider the vector fields
\[
v(p,q) = q-h'(p),~~w(p,q) = -p + (h')^{-1}(q),
\]
describing the evolution of $p$ for the Fokker-Planck dynamics (\ref{eq-minusgr-1}) and for the contact dynamics (\ref{24jul1612}),
respectively.

Denote by
$$
\pi_{pq}:
\J^1\R^n \to T^*\R^n,\; (p,q,z) \mapsto (p,q),$$
the natural projection.
Note that the point $p= (h')^{-1}(q)$
is a global attracting equilibrium point for both $v$ and $w$.

\begin{prop}\label{prop-Prig}
The Euclidean scalar product of $v$ and $w$ is strictly positive away of the equilibrium point: for some $c>0$,
\begin{equation}\label{eq-scalprod}
\left(v(p,q),w(p,q)\right) \geq c|p-(h')^{-1}(q)|^2 \;\;\forall p \in M, q \in U\;.
\end{equation}
Here $ | \cdot | $ denotes the norm induced by the Euclidean scalar product on $\R^n$.

In particular, $w$ is a gradient-like vector field for $-H$, or equivalently, $H$
is  a  global Lyapunov function for $w$ (for
 a  fixed $q$, with respect to  the  $p$-variable).
\end{prop}

\begin{proof}
Put $P:= (h')^{-1}(q)$. Since $h$ is
 a strictly convex function on the compact domain $M$,
we have
(see \cite[p. 72]{BNO})
 $$\left(v(p,q),w(p,q)\right) = \left(h'(P)-h'(p), P-p\right) \geq c|P-p|^2\;,$$
where $c$ depends only on $h$.
\end{proof}

\begin{rem}\label{rem-haslach} {\rm According to Haslach \cite{Haslach},
all macroscopic states of a thermodynamics system lie in a
certain hypersurface in the thermodynamic phase space determined by the energy of the system.
In the context of our example this hypersurface is given by (up to adjusting notations and sign conventions) $\Sigma= \{z= qp - h(p)\}$. Every path $(p(t),q(t))$ on $\R^{2n}$ can be uniquely lifted to $\Sigma$ as $$\gamma(t) = \left(p(t),q(t), q(t)p(t) - h(p(t)\right))\;.$$
Write $\lambda= dz-pdq$ for the Gibbs form.
Haslach calls such a path $\gamma(t)$ {\it admissible}
if $\lambda(\dot{\gamma}) >0$
(note again the difference between our sign convention and the one in \cite{Haslach}).
Take any trajectory $p(t)$ of the vector field $w(p,q)$
away of the equilibrium, and lift the path $(p(t),q)$ to a path $\gamma(t) \subset \Sigma$. With this language,
Proposition \ref{prop-Prig} states that $\gamma(t)$ is admissible in the sense of Haslach:
$$\lambda(\dot{\gamma}) = \dot{z}= q\dot{p} - h'(p)\dot{p} = (v(q,p), \dot{p})= (v,w) > 0\;.$$
The equality $\lambda(\dot{\gamma}) =  (v,\dot{p})$ corresponds to formula (8) in \cite{Haslach},
and our Fokker-Planck drift~$v$ corresponds to the gradient of the generalized energy in \cite{Haslach}.
The integral curves of $v$ are also admissible in the sense of Haslach, as observed in \cite{Haslach}.
}
\end{rem}

Let us recall that describing relaxation in non-equilibrium thermodynamics
by means of contact Hamiltonian flows is motivated by the fact that these flows preserve the kernel of the Gibbs form $dz-pdq$, which manifests a combination of the first and the second laws of thermodynamics. While it is  unexplored,
 with rare exceptions \cite{EP,Goto-JMP2015,Goto},
how to model such processes on the microscopic level, the above consideration shows that
certain contact Hamiltonian flows capture some features of the models coming from
a microscopic perspective: the contact flow we just described has the thermodynamic
Hamiltonian as a global Lyapunov function.
Prigogine in his 1977 Nobel lecture \cite{Pri} observes local stability of
the thermodynamic equilibrium due to the fact that thermodynamic potentials
(such as free energy) serve as Lyapunov functions near the equilibrium, and addresses a question:
{\it Does stability property hold further away from equilibrium?} Thus, Proposition~\ref{prop-Prig} is
in agreement with Prigogine's observation. At the same time, in practice one often deals with non-convex Hamiltonians
(e.g., in the case of the Curie-Weiss model considered above), in which case the contact Hamiltonian $-pq+h^*(q)$ is not smooth and the corresponding contact dynamics is not well defined. It could happen, however, that while the Hamiltonian system is ill-posed, individual trajectories with given (possibly, asymptotic) boundary conditions can be detected by means of modern symplectic topology, cf. \cite{EP}.  This is a subject of further investigation.

We refer to a recent paper \cite{Goto24} for a different take on the connection between
the Fokker-Planck equation and contact dynamics.

\section{Conclusion} \label{sec:conlcusion}
We discussed, in the context of contact thermodynamics, relaxation of the Curie-Weiss model with $N$ spins towards the equilibrium. Geometrically,
the set of equilibrium macroscopic states is given by a Legendrian submanifold in the thermodynamic phase space. We focused on a relaxation process modeled
by the gradient flow of the generating function of the equilibrium Legendrian. This generating function is given by the free energy, and the gradient equation with respect to the Wasserstein metric on the space of probability distribution gives rise to a discrete version of the Fokker-Planck equation. We showed that for
large $N$ this equation is closely related to another model of relaxation given by the Glauber Markov chain (see Section \ref{sec-comp}
and Theorem \ref{thm-drifts}). To this end, we explored the Legendrian submanifold arising in the thermodynamic limit $N \to +\infty$, and found,
by using some tools of non-smooth convex analysis, that it is given by a piece-wise smooth continuous curve consisting of branches of stable equilibria and a linear segment (see Theorem \ref{thm-limit}). Finally, in Section \ref{subsec-chords}  we looked at relaxation of the Curie-Weiss magnet after a jump of its temperature and the external magnetic field. We found that in the thermodynamic limit and for special values of parameters
there exist ``instant" relaxation processes which correspond to Reeb chords connecting
the initial and the terminal equilibrium Legendrians.

Our results lead to a number of open problems for future research.

\medskip\noindent {\bf 1.} There is an evidence that in the context of contact thermodynamics, the thermodynamic limit is related to the $\gamma$-convergence of Lagrangian submanifolds (see Remark \ref{rem-conv}). It would be interesting to
adjust the definition of the $\gamma$-convergence to thermodynamics and to prove a rigorous statement in this direction.

\medskip\noindent {\bf 2.} While we proved that the discrete Fokker-Planck and the Glauber equations are related, the link between corresponding time evolutions requires further exploration
(see Remark~\ref{rem-eq-vs-evolution}).

\medskip\noindent {\bf 3.} It would be interesting to explain the metastable behaviour of the Glauber dynamics from the geometric viewpoint (see Section \ref{subsec-meta}). The difficulty is that the metastable equilibria arise in our context in a non-direct way: one has to pass to the thermodynamic limit $N \to +\infty$, and then perform the analytic continuation of the
above-mentioned branches of stable equilibria on the limiting Legendrian submanifold.

\medskip\noindent {\bf 4.}  It would be interesting to relate the Fokker-Plank equation to the
Cauchy-Riemann equation as both arise as gradient equations for generating functions in symplectic topology (see Remark \ref{rem-Floer}).

\medskip\noindent {\bf 5.} It would be interesting to explore further ``instant" relaxation
processes corresponding to Reeb chords observed in Section \ref{subsec-chords}, and, in particular, detect such chords in more sophisticated examples by methods of contact topology. Furthermore,
it is unclear what is (if any) a counterpart of these processes for the case of finite $N$.

\section{Appendix: Calculating the gradient \eqref{eq-grad-Wass}} \label{App-2}

For vectors $v,w \in \R^{N+1}$ put
$$\langle v,w \rangle = \frac{2}{N}\sum_i v_iw_i\;,$$
With this notation one readily checks that
\begin{equation}\label{ap2-1}
\langle \phi, \Div \Psi \rangle  =  -  (\nabla \phi, \Psi)\;,
\end{equation}
for every vector $\phi$ and antisymmetric matrix $\Psi$.
Furthermore,
\begin{equation}\label{ap2-2}
\Div \nabla \rho = \frac{N^2}{4}(K - {\bf 1})\rho\;.
\end{equation}
Put $H=(H^{(1)},\dots,H^{(N+1)})$ with
\be\label{24jul1508}
H^{(i)}:= N^{-1}H_N(m_i),
\ee
where
$H_N$ is given by \eqref{eq-ham-red}. Observe that
\begin{equation}\label {ap2-3}
\overline{\Phi}_N = -\frac{1}{\beta N} \langle \rho, \ln \rho \rangle - \langle \rho, H \rangle\;.
\end{equation}
Denote $v := \dot{\rho} \in T_\rho {\mathcal P}$ and write it in the form $v = - \Div (\hat{\rho}\odot \nabla \phi)$.

First, put $G:= \langle \rho, H \rangle$, and
$$\xi = - \Div (\hat{\rho} \odot \nabla H)\;.$$
Then, by \eqref{ap2-1},  the differential of $G$ evaluated on the tangent vector $v$
is given by
$$dG(v) = \langle v, H \rangle = -\langle\Div (\hat{\rho}\odot \nabla \phi), H\rangle = (\hat{\rho}\odot \nabla \phi,  \nabla H) = (\xi, v)_W\;.$$
Thus,
\begin{equation}\label{eq-ap2-10}
\nabla_W G = - \Div (\hat{\rho} \odot \nabla H)\;.
\end{equation}

Next, put $E := \langle \rho, \ln \rho \rangle$, and note that
$$\hat{\rho}_{ij} \cdot (\nabla \ln\rho)_{ij} = (\nabla \rho)_{ij}\;.$$
It follows that
$$dE(v)= \langle v, \ln \rho \rangle + {2}{N}\sum_i \rho_i \cdot \frac{v_i}{\rho_i}\;.$$
Recall that $\sum v_i = 0$ for every $v \in T_\rho \mathcal{P}$, and hence the second summand vanishes. Thus,
$$dE(v)= (\nabla \phi,  \hat{\rho}\odot \nabla(\ln\rho)) =
(\nabla \phi, \nabla \rho)\;.$$
Choose $\psi$ such that $\nabla \rho = \hat{\rho}\odot \nabla \psi$.
By \eqref{ap2-2},
$$\eta:= -\Div( \hat{\rho}\odot \nabla \psi) = -\frac{N^2}{4}(K - {\bf 1}) \rho. $$
Thus,
$$dE(v) = (v,\eta)_W\;,$$
and hence
\begin{equation}\label{eq-ap2-11}
\nabla_W E =-\frac{N^2}{4}(K - {\bf 1}) \rho \;.
\end{equation}
Since by \eqref{ap2-3}
$$\overline{\Phi} =  -\frac{1}{\beta N} E - H\;,$$
we get from \eqref{eq-ap2-10} and \eqref{eq-ap2-11}
$$\nabla_W \overline{\Phi} (\rho) =   \frac{N}{4\beta} (K - {\bf 1}) \rho + \Div (\hat{\rho} \odot \nabla H)\;.$$
Recalling (\ref{eq-ham-red}) and (\ref{24jul1508}), we see that we have derived formula \eqref{eq-grad-Wass}.

\section{Appendix: Calculating transition probabilities \eqref{eq-trans}}
We use the notation of Section \ref{finite-CW} and Section \ref{sec-comp}.
In \cite{Glauber} Glauber considered a Markov chain represented by the following graph:
the space of vertices is the space of spin configurations $\Upsilon_N$, and two configurations $\sigma,\sigma'$ are connected by an edge whenever they differ at exactly one point
$j \in \{1, \dots, N\}$. Recall that $\theta(m) = \tanh (\beta(q+bm))$.
Fix $q$. Recall that $m = N^{-1}\sum \sigma_i$. Put $\Delta = \frac{1}{2}(h(\sigma')-h(\sigma))$,
where $h$ is given by \eqref{eq-hamnoneff}.
The transitions happen with the time step $\tau$.
The transition probabilities $P(\sigma,\sigma')$ are calculated based on the following
standard assumptions.
First, we assume the detailed balance condition
$$\frac{P(\sigma,\sigma')}{P(\sigma',\sigma)} = e^{-2\beta\Delta}\;.$$
Second, we assume that the transition rate at each edge is constant:
$$P(\sigma,\sigma') + P(\sigma',\sigma) = 4c\tau\;,$$
where $4c$ is a time scale parameter.
This yields
\begin{equation}\label{eq-prob-1}
P(\sigma,\sigma') = 2c\tau (1 - \tanh (\beta\Delta))\;.
\end{equation}
A direct calculation shows that if $\sigma_j=-\sigma'_j=1$, then $\Delta = q+b(m-2/N)$,
and if $\sigma_j = -\sigma'_j=-1$, then $\Delta = -q - b(m+2/N)$.
Therefore,
\begin{equation} \label{eq-tran-pm}
P(\sigma,\sigma') ={2c\tau}(1 - \theta(m -2/N)), \;\;\text{if}\;\; \sigma_j=1, \sigma'_j=-1\;,
\end{equation}
and
\begin{equation} \label{eq-tran-mp}
P(\sigma,\sigma') = {2c\tau}(1 + \theta(m +2/N)), \;\;\text{if}\;\; \sigma_j=-1, \sigma'_j=1\;.
\end{equation}

Next we pass from the Markov chain on $\Upsilon_N$ to the one on $M_N$.
We use the approach of Sections 13.1 and 9.3 of \cite{BH}.
The {\it lumping} construction enables one to replace the Glauber Markov chain
by the random walk on the set
$$M_N = \{-1, -1+2/N, \dots, 1-2/N, 1\}\;.$$
Write $\Pi_-(m)$ for the probability
of jump from $m$ to $m-2/N$, and $\Pi_+(m)$ for the probability
of jump from $m$ to $m+2/N$. These probabilities can be computed
by using
formula (9.3.3) in \cite{BH}. Namely, choose $\sigma \in \Upsilon_N$ with $m(\sigma)= m$.
Then $$\Pi_-(m) = \sum_{\sigma'}P(\sigma,\sigma')\;,$$
where the sum is taken over all spin configurations $\sigma'$ with $m(\sigma')= m-2/N$ and which are connected with
$\sigma$ by an edge. Since $\sigma$ has exactly $N(1+m)/2$ slots occupied by $+1$ and
each of them can flip to $-1$, we have from \eqref{eq-tran-pm} that
$$\Pi_-(m) = cN\tau(1+m)(1 - \theta(m -2/N))\;.$$
This proves the first equality in \eqref{eq-trans}. The proof of the second one is analogous.
\qed

\section{Appendix: Calculation of the effective Hamiltonian}  \label{sec-entropy}
Consider a binomial coefficient
$$C_N(m) = \binom{N}{\frac{(1+m)}{2} N}\;.$$
Let $\rho(m)$ be a probability density on $M_N$, and $P(\sigma)$ be the lift of the measure
$\rho d\mu$ to $\Upsilon_N$:
\begin{equation}\label{eq-lift}
\frac{2}{N}\rho (m(\sigma)) = C_N (m(\sigma))P(\sigma)\;.
\end{equation}
The entropy $S(P)$ is given by
$$S(P) = -\sum_\sigma P(\sigma)\ln P(\sigma)\;.$$
By using \eqref{eq-lift} we calculate
\begin{equation}\label{eq-S-1}
S(P) = - \int_{M_N} \rho \ln \rho d\mu - \ln(2/N) + \int_{M_N} \ln C_N(m) \rho  d\mu\;.
\end{equation}
By using \cite[Corollary 2.4]{St} we get that
\begin{equation}\label{eq-binomial-1}
C_N(m) = -Nc(m) -\frac{1}{2}\ln(2\pi N)+ \frac{1}{2}\ln\frac{4}{1-m^2}
 + \frac{1}{12N}\left(1-\frac{4}{1-m^2}\right) + O(1/N^2)\;,
 \end{equation}
with
$$c(m) = \frac{1+m}{2} \ln \frac{1+m}{2} + \frac{1-m}{2}\ln \frac{1-m}{2}\;.$$
Recall that the Hamiltonian $h$ of a spin configuration is given by \eqref{eq-hamnoneff}.
Thus, by \eqref{eq-Phi-rho},
\begin{equation}\label{eq-ent-2}
\Phi(q,\rho) = -\beta^{-1} S(\rho) + N\int_{M_N}( h(q,m) + \beta^{-1}c(m))\rho d\mu + \int \beta^{-1}r_N(m)\rho(m)d\mu\;,
\end{equation}
with
\begin{equation}\label{eq-ent-3}
r_N(m) = \ln(2/N) +  \frac{1}{2}\ln(2\pi N)- \frac{1}{2}\ln\frac{4}{1-m^2}
 - \frac{1}{12N}\left(1-\frac{4}{1-m^2}\right) + O(1/N^2)\;.
 \end{equation}
Note that $h + \beta^{-1}c = F_{b,\beta}$, which proves \eqref{eq-ham-red}.
Expression \eqref{eq-ent-3} is identical to \eqref{eq-remainder}.
\qed

\noindent {\sc Michael Entov\\
Department of Mathematics,\\
Technion - Israel Institute of Technology,\\
Haifa 32000, Israel}\\
{\tt  entov@technion.ac.il}
\\

\noindent {\sc Leonid Polterovich\\
School of Mathematical Sciences,\\
Tel Aviv University,\\
Tel Aviv 69978, Israel}\\
{\tt  polterov@tauex.tau.ac.il}
\\

\noindent {\sc Lenya Ryzhik\\
Department of Mathematics,\\
Stanford University,\\
Stanford CA 94305, USA}\\
{\tt  ryzhik@stanford.edu}
\\


\begin{thebibliography}{99}

\bibitem{Be} Berezin, F., \emph{Lectures on statistical physics, 1966-67}, translated from Russian and edited by D. Leites,  
    2009.

\bibitem{BNO}  Bertsekas, D., Nedic, A., Ozdaglar, A., Convex analysis and optimization. Athena Scientific, 2003.

\bibitem{BH} Bovier, A., den Hollander, F., Metastability. A potential-theoretic approach. Springer, Cham, 2015.

\bibitem{BLN} Bravetti, A., Lopez-Monsalvo, C.S., Nettel, F., \emph{Contact symmetries and Hamiltonian thermodynamics},  Annals of Physics, \textbf{361} (2015), 377-400.

    \bibitem{DRS} Dimitroglou Rizell, G., Sullivan, M.,  $ C^ 0$-limits of Legendrians and positive loops. arXiv preprint arXiv:2212.09190. 2022 Dec 18.

\bibitem{EFLS} Erbar. M., Fathi, M., Laschos, V., Schlichting, A., \emph{Gradient flow structure for McKean-Vlasov equations on discrete spaces}, Discrete and Continuous Dynamical Systems (Series A) \textbf{36} (2016), 6799-6833.

\bibitem{EK} Ermolaev, V., Külske, C., \emph{Low-temperature dynamics of the Curie-Weiss model: Periodic orbits, multiple histories, and loss of Gibbsianness}, Journal of Statistical Physics \textbf{141} (2010), 727-756.


\bibitem{EP} Entov M., Polterovich L., \emph{Contact topology and non-equilibrium thermodynamics}, Nonlinearity \textbf{36} (2023), 3349-3375.

\bibitem{FV} Friedli, S., Velenik, Y., Statistical mechanics of lattice systems. A concrete mathematical introduction. Cambridge Univ. Press, Cambridge, 2018.

\bibitem{Glauber} Glauber, R.J., \emph{Time-dependent statistics of the Ising model}, J. Math. Phys. \textbf{4} (1963), 294-307.

\bibitem{Goto-JMP2015} Goto, S., \emph{Legendre submanifolds in contact manifolds as attractors and geometric non-equilibrium thermodynamics},
J. Math. Phys. \textbf{56} (2015), 073301.

\bibitem{Goto} Goto, S., \emph{Nonequilibrium thermodynamic process with hysteresis and metastable states -- a contact Hamiltonian with unstable and stable segments of a Legendre submanifold}, J. Math. Phys. \textbf{63} (2022), Paper No. 053302, 25 pp.

    \bibitem{Goto24} Goto, S., \emph{From the Fokker-Planck equation to a contact Hamiltonian system}, Journal of Physics A: Mathematical and Theoretical. 2024 Mar 12.


\bibitem{GLP} Goto, S., Lerer, S., Polterovich, L., \emph{Contact geometric approach to Glauber dynamics near a cusp and its limitation},
J. Phys. A \textbf{56} (2023), Paper No. 125001, 16 pp.

\bibitem{GrO}
Grmela, M., \"Ottinger, H.,
\emph{Dynamics and thermodynamics of complex fluids. I. Development of a general formalism},
Phys. Rev. E., {\bf 56}, (1997), 6620-6632.


\bibitem{Grmela} Grmela, M., \emph{Contact geometry of mesoscopic thermodynamics and dynamics}, Entropy \textbf{16} (2014), 1652-1686.

\bibitem{Haslach} Haslach Jr, H.W.,  \emph{Geometric structure of the non-equilibrium thermodynamics of homogeneous systems}, Reports on Mathematical Physics, \textbf{39} (1997), 147-162.


\bibitem{Her} Hermann, R., Geometry, physics, and systems. Marcel Dekker Inc., New York, 1973.


\bibitem{IS} Ishii, H., Souganidis, P.E., \emph{Metastability for parabolic equations with drift: Part I}, Indiana Univ. Math. J. \textbf{64} (2015), 875-913.

\bibitem{JKO} Jordan, R., Kinderlehrer, D., Otto, F., \emph{The variational formulation of the Fokker-Planck equation}, SIAM J. Math. Anal. \textbf{29} (1998), 1-17.

\bibitem{J} Jordan, R., Kinderlehrer, D., Otto, F., \emph{Dynamics of the Fokker-Planck equation}, Phase Transitions \textbf{69} (1999), 271-288.


\bibitem{KPW} Kochma\'{n}ski, M., Paszkiewicz, T., Wolski, S., \emph{Curie–Weiss magnet -- a simple model of phase transition}, Eur. J. of Phys. \textbf{34} (2013) 1555-1573.

\bibitem{KH} Koper, G.J., Hilhorst, H.J., \emph{Nonequilibrium dynamics and aging in a one-dimensional Ising spin glass},
    Physica A: Statistical Mechanics and its Applications \textbf{155} (1989), 431-459.


\bibitem{LimOh} Lim, J., Oh, Y.-G., \emph{Nonequilibrium thermodynamics as a symplecto-contact reduction and relative information entropy}, Reports on Mathematical Physics \textbf{92} (2023), 347-400.

\bibitem{Lott} Lott, J., \emph{Some Geometric Calculations on Wasserstein Space}, Comm. in Math. Phys.  \textbf{277} (2008), 423-437.


\bibitem{Maas} Maas, J., \emph{Gradient flows of the entropy for finite Markov chains}, J. Funct. Anal. \textbf{261} (2011), 2250-2292.

\bibitem{Mussardo} Mussardo, G., Statistical field theory: an introduction to exactly solved models in statistical physics. Oxford Univ. Press, 2010.

\bibitem{Oh} Oh, Y.-G., \emph{Symplectic topology as the geometry of action functional. I. Relative Floer theory on the cotangent bundle}, J. Diff. Geom. \textbf{46} (1997), 499-577.


\bibitem{Otto} Otto, F., \emph{The geometry of dissipative evolution equations: the porous medium equation},
Comm. Part. Diff. Equations \textbf{26} (2001), 101-174.

\bibitem{Penrose} Penrose, O., \emph{Metastable decay rates, asymptotic expansions, and analytic continuation of thermodynamic functions},  Journal of Statistical Physics \textbf{78} (1995), 267-283.

\bibitem{Pri} Prigogine, I., \emph{Time, structure, and fluctuations},
     Science \textbf{201}:4358 (1978),  777-785.

\bibitem{RW} Rockafellar, R.T., Wets, R.J.-B., Variational analysis. Springer-Verlag, Berlin, 1998.


\bibitem{Sto} Stoki\'{c}, M., $ C^ 0$-flexibility of Legendrian discs in $\mathbb {R}^ 5$. arXiv preprint arXiv:2406.04194. 2024 Jun 6.

\bibitem{St} St\u{a}nic\u{a}, P., \emph{Good lower and upper bounds on binomial coefficients}, JIPAM. J. Inequal. Pure Appl. Math. \textbf{2} (2001), Article 30.



\bibitem{Suzuki-Kubo}   Suzuki, M., Kubo, R., \emph{Dynamics of the Ising model near the critical point. I.},  Journal of the Physical Society of Japan \textbf{24} (1968), 51-60.


\bibitem{V} van der Schaft, A., Maschke, B., \emph{Geometry of thermodynamic processes}, Entropy \textbf{20} (2018), 925.

\bibitem{Vi} Viterbo, C., \emph{On the supports in the Humili\'{e}re completion and $\gamma$-coisotropic sets (with an Appendix joint with Vincent Humili\'{e}re)}, preprint, arXiv:2204.04133, 2022.

\end{thebibliography}
\end{document}